\documentclass[a4paper]{article}

\usepackage[margin=1in]{geometry}
\usepackage{graphicx}%
\usepackage{multirow}%
\usepackage{amsmath,amssymb,amsfonts}%
\usepackage{amsthm}%
\usepackage{mathrsfs}%
\usepackage[title]{appendix}%
\usepackage{xcolor}%
\usepackage{textcomp}%
\usepackage{manyfoot}%
\usepackage{booktabs}%
\usepackage{algorithm}%
\usepackage{algorithmicx}%
\usepackage{algpseudocode}%
\usepackage{listings}%
\usepackage{tikz-cd}
\usepackage{lineno}
\tikzcdset{every label/.append style = {font = \small}}
\usepackage{tikz}
\usetikzlibrary{graphs,decorations.pathmorphing,decorations.markings}
\usepackage{subcaption}
\usepackage{url}
\usepackage{hyperref}
\usepackage{mathrsfs}
\usepackage{color}
\usepackage{bbm}

\newtheorem{theorem}{Theorem}%  meant for continuous numbers
\newtheorem{example}{Example}%
\newtheorem*{theorem*}{Theorem}

\newtheorem{definition}{Definition}%

\newtheorem{lemma}[theorem]{Lemma}

\usetikzlibrary{graphs,decorations.pathmorphing,decorations.markings}
\usepackage{multirow}

\numberwithin{equation}{section}

\begin{document}

\title{The Boundary Reproduction Number for Determining Boundary Steady State Stability in Chemical Reaction Systems}

\author{%
  Matthew D. Johnston\thanks{Department of Mathematics and Computer Science, Lawrence Technological University, 21000 W 10 Mile Rd, Southfield, MI 48075, USA, \tt{${}^{*}$mjohnsto1@ltu.edu}} \and Florin Avram\thanks{Laboratoire de Math\'{e}matiques Appliqu\'{e}es, Universit\'{e} de Pau, 64000, Pau, France}\\ \\%, %\corrauth
%  Bruce Pell, 
%  Jared Pemberton, and  David A. Rubel \\ \\ 
%    Department of Mathematics + Computer Science\\Lawrence Technological University\\21000 W 10 Mile Rd, Southfield, MI 48075, USA\\ \tt{${}^{*}$mjohnsto1@ltu.edu}
}

% \shortauthors is used in copyright information in the end of the paper
%\shortauthors{MDJ, BP, and DAR}

%\address{%
%  \addr{}{}
%}

% corresponding author
%\corraddr{}

\maketitle

\begin{abstract}
We introduce the boundary reproduction number, adapted from the next generation matrix method, to assess whether an infusion of species will persist or become exhausted in a chemical reaction system. Our main contributions are as follows: (a) we show how the concept of a siphon, prevalent in Petri nets and chemical reaction network theory, identifies sets of species that may become depleted at steady state, analogous to a disease-free boundary steady state; (b) we develop an approach for incorporating biochemically motivated conservation laws, which allows the stability of boundary steady states to be determined within specific compatibility classes; and (c) we present an effective heuristic for decomposing the Jacobian of the system that reduces the computational complexity required to compute the stability domain of a boundary steady state. The boundary reproduction number approach significantly simplifies existing parameter-dependent methods for determining the stability of boundary steady states in chemical reaction systems and has implications for the capacity of critical metabolites and substrates in metabolic pathways to become exhausted.
\end{abstract}

%\keywords{chemical reaction network, basic reproduction number, next-generation matrix, epidemiology}

%\tableofcontents

\section{Introduction}
\label{sec:introduction}

There has been significant research interest in the areas of mathematical biochemistry and mathematical epidemiology in the capacity of species to tend toward zero \cite{martcheva2015,feinberg2019foundations}. In mathematical biochemistry, a chemical species tending to zero may correspond to a catalyst or metabolite being depleted, potentially shutting down a critical metabolic pathway. In mathematical epidemiology, a compartment tending to zero may indicate the extinction of an infectious disease, such as a virus or bacterial infection, and is therefore crucial to our understanding of disease management. Despite the importance of the stability of boundary steady states in both disciplines, and the significant similarities in the modeling frameworks, little work has been done to extend techniques from one area to the other. % \cite{avram2024investigating}. 

In recent years, researchers in chemical reaction network theory have focused on identifying \emph{parameter-independent structural properties} of the underlying interaction network that guarantee persistence and permanence in chemical reaction systems—that is, conditions ensuring that no initially positive trajectory approaches the boundary \cite{craciun2013persistence,brunner2018robust,craciun2019polynomial}. Examples of well-studied structural properties related to persistence include complex-balancing \cite{horn1972general,anderson2008global,craciun2019polynomial}, weak reversibility \cite{anderson2010dynamics,anderson2011proof,deng2011steady,august2010solutions}, dynamical nonemptiability \cite{angeli2007petri,johnston2011weak}, endotacticity \cite{craciun2013persistence,gopalkrishnan2014geometric,johnston2016computational,craciun2020endotactic,anderson2020tier}, and monotonicity \cite{angeli2006structural,angeli2008translation,leenheer2007monotone,wang2008singularly}. These studies aim to establish conditions under which persistence or permanence holds \emph{for all parameter values} and all \emph{stoichiometric compatibility classes}; they do not address systems for which the boundary is attracting for some parameter values or compatibility classes but repelling for others.

Research in mathematical epidemiology, conversely, has largely focused on characterizing the \emph{parameter-dependent} \emph{basic reproduction number}, $\mathscr{R}_0$, which is the expected number of new cases generated by a typical infected individual in a fully susceptible population \cite{diekmann1990definition}. The basic reproduction number is straightforward to interpret: if $\mathscr{R}_0 > 1$, the disease will spread; if $\mathscr{R}_0 < 1$, the disease will die out. Mathematically, this represents a bifurcation point at which the stability of a boundary steady state changes. Mathematical research has focused on methods for efficiently computing $\mathscr{R}_0$, such as the next generation matrix method \cite{VANDENDRIESSCHE2002,Heffernan2005,van2008further}. This method has been applied to a variety of disease models, including multi-strain models \cite{lazebnik2022generic}, vector-borne disease models \cite{van2017reproduction,van2017basic}, and models with vaccination \cite{ehrhardt2019sir,johnston2025effect}.

In this paper, we introduce the boundary reproduction number of a boundary steady state in mathematical biochemistry. This concept extends and formalizes recent work in \cite{Avram2025StabilityIR}, which heuristically applied the next generation method from mathematical epidemiology to biochemical examples. The boundary reproduction number approach provides a significant computational advantage for determining parameter-dependent stability thresholds for boundary steady states, as compared to existing direct methods such as computing the eigenvalues of the Jacobian or applying the Routh-Hurwitz conditions \cite{routh1877treatise,hurwitz1895ueber}. We present three results supporting the boundary reproduction number framework: a direct method (Theorem \ref{thm:main}), a matrix-based method (Theorem \ref{thm:2}), and a network-structure-based method (Theorem \ref{thm:3}).

To illustrate the boundary reproduction number method, consider the following network (left) and corresponding mass-action system (right), simplified from an EnvZ-OmpR network presented in \cite{shinar2010structural}:
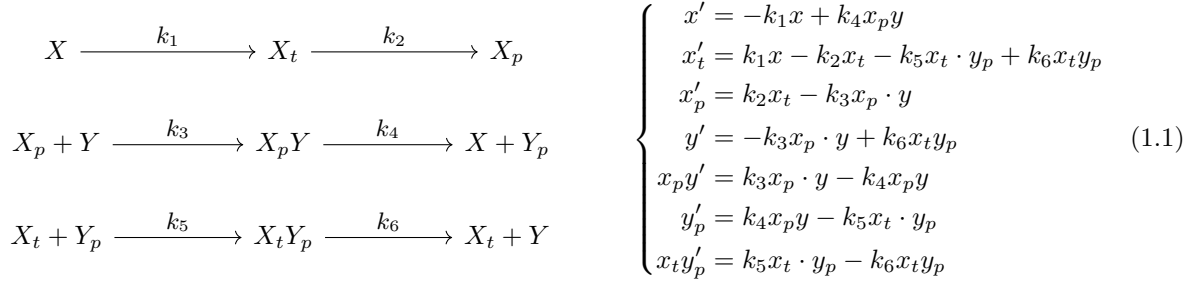
\begin{equation}
\label{EnvZOmpr-simplified}
\begin{tikzcd}
X \arrow[rr,"k_1"] & & X_t \arrow[rr,"k_2"] & & X_p \\
X_p + Y \arrow[rr,"k_3"] & & X_pY \arrow[rr,"k_4"] & & X + Y_p \\
X_t + Y_p \arrow[rr,"k_5"] & & X_tY_p \arrow[rr,"k_6"] & &  X_t + Y
\end{tikzcd} \; \; \; \; \; \; \; \; \;
\left\{ \begin{aligned}
x' & = -k_1x+k_4x_py \\
x_t' & = k_1x-k_2x_t - k_5 x_t \cdot y_p + k_6 x_ty_p\\
x_p' & = k_2x_t-k_3x_p \cdot y\\
y' & = -k_3x_p \cdot y+k_6x_ty_p\\
x_py' & = k_3 x_p \cdot y - k_4 x_py\\
y_p' & = k_4x_py-k_5x_t \cdot y_p \\
x_ty_p' & = k_5 x_t \cdot y_p - k_6 x_ty_p
\end{aligned}\right.
\end{equation}
Note that this system cannot be interpreted as a model of infectious disease spread because it includes synthesis (e.g., $X_1 + X_2 \to X_3$) and disassociation reactions (e.g., $X_1 \to X_2 + X_3$). %Such reactions are common in biochemistry but violate the population compartmentalization assumptions required in epidemiological models, specifically the assumption that, aside from births and deaths, the loss (respectively, gain) of one individual in a compartment must be offset by the gain (respectively, loss) of one individual in another.
Nevertheless, the boundary reproduction number method can still be applied. The network \eqref{EnvZOmpr-simplified} has a \emph{critical siphon} $\mathcal{X} = \{ X, X_t, Y, X_pY, X_tY_p \}$ corresponding to the set of boundary steady states $\mathbf{x}^* = (x, x_t, x_p, y, x_py, y_p, x_ty_p) = (0, 0, X_{tot}, 0, 0, Y_{tot}, 0)$. The boundary reproduction number of $\mathbf{x}^*$ can be computed as $\displaystyle \mathscr{R}_{\mathbf{x}^*} = \rho(FV^{-1}) = k_5Y_{tot}/k_2$, which depends on both the rate parameters $k_2$ and $k_5$ and the conservation constant $Y_{tot}$. We show that the boundary steady state is unique within each stoichiometric compatibility class, is locally asymptotically stable within that class if $\displaystyle \mathscr{R}_{\mathbf{x}^*} < 1$ (i.e., $\displaystyle Y_{tot} < k_2/k_5$), and unstable within that class if $\mathscr{R}_{\mathbf{x}^*} > 1$ (i.e., $\displaystyle Y_{tot} > k_2/k_5$).

\section{Background}
\label{sec:background}

In this section, we introduce the background knowledge on mathematical models of chemical reaction networks.

\subsection{Chemical Reaction Networks}
\label{sec:crn}

The following concept is the foundational object of study in this paper \cite{horn1972general,horn1972necessary,feinberg1972chemical,feinberg2019foundations}.

\begin{definition}
\label{def:crn}
A \textbf{chemical reaction network} $(\mathcal{S}, \mathcal{R})$ consists of a species set $\mathcal{S} = \{ X_1, \ldots, X_n \}$ together with a reaction set $\mathcal{R} = \{ R_1, \ldots, R_r \}$. The elementary reactions have the form
\[R_j: \; \; \;  \alpha_{1j} X_1 + \alpha_{2j} X_2 + \cdots + \alpha_{nj} X_n \longrightarrow \beta_{1j} X_1 + \beta_{2j} X_2 + \cdots + \beta_{nj} X_n, \; \; \; j=1, \ldots, r\]
where the constants $\alpha_{ij}, \beta_{ij} \in \mathbb{Z}$ are called \textbf{stoichiometric constants}. We say that $X_i \in \mathcal{S}$ is a \textbf{reactant} of reaction $R_j$ if $\alpha_{ij} > 0$ and a \textbf{product} of reaction $R_j$ if $\beta_{ij} > 0$. Given $\alpha_j = ( \alpha_{1j}, \alpha_{2j}, \ldots, \alpha_{nj} ) \in \mathbb{Z}^n_{\geq 0}$ and $\beta_j = ( \beta_{1j}, \beta_{2j}, \ldots, \beta_{nj} ) \in \mathbb{Z}^n_{\geq 0}$, the vector $\beta_j - \alpha_j \in \mathbb{Z}^n$ is called the \textbf{reaction vector} of reaction $R_j$. We define the \textbf{stoichiometric matrix} $\Gamma \in \mathbb{Z}^{n \times r}$ to be the matrix with entries $\Gamma_{ij} = \beta_{ij}-\alpha_{ij}$ (i.e., columns given by the reaction vectors $\beta_j - \alpha_j$), and the \textbf{stoichiometric subspace} to be $S = \text{im}(\Gamma)$. The dimension of $S$ is denoted $s = \text{dim}(S)$.
\end{definition}

Given a chemical reaction network $(\mathcal{S},\mathcal{R})$, we associate the \emph{chemical reaction system}
\begin{equation}
\label{de}
\frac{d\mathbf{x}}{dt} = \Gamma R(\mathbf{x})
\end{equation}
where $\mathbf{x} = (x_1, \ldots, x_n) \in \mathbb{R}_{\geq 0}^n$ is the vector of species concentrations and $R(\mathbf{x}) = (R_1(\mathbf{x}), \ldots, R_r(\mathbf{x})) \in \mathbb{R}_{\geq 0}^r$ is the reaction rate vector. We assume that the reaction rates $R_j(\mathbf{x})$ satisfy the following regularity assumptions:
\begin{enumerate}
\item[\textbf{(A1)}] Reaction rates depend smoothly on the concentrations (i.e., $R_j(\mathbf{x}) \in \mathscr{C}^1$);
\item[\textbf{(A2)}] Reaction rates are nonnegative and strictly positive if and only if all reactants are present (i.e., $R_j(\mathbf{x}) \geq 0$ with strict positivity if and only if $x_j > 0$ for all $\alpha_{ij} > 0$); and
\item[\textbf{(A3)}] Reaction rates are nondecreasing with respect to their reactants (i.e., $\displaystyle{\frac{\partial R_j(\mathbf{x})}{\partial x_i} \geq 0}$ if $\alpha_{ij} > 0$).
\end{enumerate}

A classical kinetic assumption consistent with \textbf{(A1)}–\textbf{(A3)} is the \emph{law of mass-action}, which states that the rate of a reaction is proportional to the product of the reactant concentrations \cite{guldberg1879uber}. For example, a reaction of the form $X_1 + 2 X_2 \stackrel{k_j}{\longrightarrow} \cdots$ has the rate $k_jx_1x_2^2$. When the reaction rates are not mass-action, we indicate the associated reaction rate above the arrow, e.g., $X_1 \stackrel{R_j(\mathbf{x})}{\longrightarrow} X_2$. When the system \eqref{de} uses mass-action kinetics, we call it a \emph{mass-action system}.

To further understand the state space of \eqref{de}, we define the following.

\begin{definition}
Consider a chemical reaction network $(\mathcal{S},\mathcal{R})$ with corresponding chemical reaction system \eqref{de} and stoichiometric subspace $S = \text{im}(\Gamma)$. Then:
\begin{enumerate}
\item The \textbf{stoichiometric compatibility class} of \eqref{de} associated with an initial condition $\mathbf{x}_0 \in \mathbb{R}_{\geq 0}^n$ is given by $\mathsf{C}_{\mathbf{x}_0} = (S + \mathbf{x}_0) \cap \mathbb{R}_{\geq 0}^n$.
\item A \textbf{conservation vector} of \eqref{de} is a vector $\mathbf{w} \in \mathbb{R}^n$ satisfying $\mathbf{w} \in S^{\perp}$ (i.e., $\mathbf{w} \in \text{null}(\Gamma^T)$). The \textbf{conservation law} associated with the conservation vector $\mathbf{w} \in \mathbb{R}_{\geq 0}^n$ is given by $\mathbf{w}^T \mathbf{x} = \Lambda$, where $\Lambda$ is called the \textbf{conservation constant} of the conservation law. Given a linearly independent set of $n-s$ conservation vectors $\{ \mathbf{w}_1, \ldots, \mathbf{w}_{n-s} \} \subset \mathbb{R}^n$, we can define the \textbf{conservation matrix} $W \in \mathbb{R}^{n \times (n-s)}$ where $W_{i,\cdot} = \mathbf{w}_i$, and a \textbf{conservation constant vector} $\mathbf{\Lambda} = (\Lambda_1, \ldots, \Lambda_{n-s}) \in \mathbb{R}^{n-s}$ so that
\begin{equation}
\label{conservation}
W \mathbf{x} = \mathbf{\Lambda}.
\end{equation}
A network is said to be \textbf{conservative} if there is a conservation vector $\mathbf{w} \in S^{\perp} \cap \mathbb{R}_{> 0}^n$.
\end{enumerate}
\end{definition}

Several important restrictions on the dynamics of \eqref{de} are known to follow from the structure of the stoichiometric compatibility classes and conservation laws. By positivity of solutions and the property $d\mathbf{x}/dt \in \text{im}(\Gamma) = S$ for \eqref{de}, if $\mathbf{x}(0)=\mathbf{x}_0$, then $\mathbf{x}(t) \in \mathsf{C}_{\mathbf{x}_0}$ for all $t \geq 0$ \cite{horn1972general,feinberg1972chemical}. That is, solutions of \eqref{de} remain in the stoichiometric compatibility class associated with the initial condition. We also have that $\mathbf{w}^T \mathbf{x}(t) = \mathbf{w}^T \mathbf{x}_0 = \Lambda$ for all $t \geq 0$. Note that there is a one-to-one correspondence between the conservation constants $\mathbf{\Lambda}$ and the stoichiometric compatibility classes $\mathsf{C}_{\mathbf{x}_0}$; that is, $\mathsf{C}_{\mathbf{x}_0}$ can be uniquely determined from $\mathbf{\Lambda}$, and vice-versa. If the network is conservative, then $\mathsf{C}_{\mathbf{x}_0}$ is a compact set, so solutions of \eqref{de} are bounded.

\subsection{Siphons}

The structure of chemical reaction networks has played a significant role in understanding whether solutions of the chemical reaction system \eqref{de} can tend toward zero. In this paper, we utilize the following concept, introduced to chemical reaction network theory in \cite{angeli2007petri} and further developed in several works \cite{shiu2010siphons,angeli2009persistence,donnell2014control}.

\begin{definition}
\label{def:siphon}
Consider a chemical reaction network $(\mathcal{S},\mathcal{R})$. A subset $\mathcal{X} \subseteq \mathcal{S}$ is a \textbf{siphon} if every reaction $R_k \in \mathcal{R}$ that has a species $X_i \in \mathcal{X}$ as a product also has a species $X_j \in \mathcal{X}$ as a reactant. We define the \textbf{antisiphon} $\mathcal{Y} = \mathcal{S} \setminus \mathcal{X}$. A siphon $\mathcal{X}$ is \textbf{noncritical} if there exists a conservation vector $\mathbf{w} \in \mathbb{R}_{\geq 0}^m$ for which $w_i > 0$ implies $X_i \in \mathcal{X}$; otherwise, the siphon is \textbf{critical}.
\end{definition}

Siphons have been studied in chemical reaction network theory due to their connection with forward-invariant boundary faces of the positive orthant $\mathbb{R}_{\geq 0}^n$ and boundary steady states. We define the following.

\begin{definition}
\label{def:invariant}
Consider a chemical reaction network $(\mathcal{S},\mathcal{R})$ with corresponding chemical reaction system \eqref{de}. Suppose $\mathcal{X} \subseteq \mathcal{S}$ is a siphon as defined in Definition \ref{def:siphon}. Then the $\mathcal{X}$\textbf{-free boundary face} of the state space $\mathbb{R}^n_{\geq 0}$ is given by
\[
B_{\mathcal{X}} = \{ \mathbf{x} \in \mathbb{R}_{\geq 0}^n \mid x_i = 0 \text{ for } X_i \in \mathcal{X} \}.
\]
The set of $\mathcal{X}$-\textbf{free boundary steady states} is given by
\[
E_{\mathcal{X}} = \{ \mathbf{x}^* \in \mathbb{R}_{\geq 0}^n \mid x_i^* = 0 \text{ for } X_i \in \mathcal{X} \text{ and } \Gamma R(\mathbf{x}^*) = \mathbf{0} \} \subseteq B_{\mathcal{X}}.
\]
\end{definition}

It is known that $B_{\mathcal{X}}$ is an invariant region of the system \eqref{de} (i.e., $\mathbf{x}(0) \in B_{\mathcal{X}}$ implies $\mathbf{x}(t) \in B_{\mathcal{X}}$ for all $t \geq 0$). It is also known that, if $\mathcal{X}$ is a noncritical siphon, then \eqref{de} cannot approach $B_{\mathcal{X}}$, a property called \emph{structural persistence} \cite{angeli2007petri}.

Computational methods for finding siphons in Petri nets require specialized packages and can be computationally intensive \cite{liu2016survey,donnell2014control}. Consequently, for many networks, it is often more practical to compute the zeros of the corresponding chemical reaction system \eqref{de}. Packages for computing the zeros of algebraic systems of equations with both symbolic and numeric parameters are broadly available.

\section{Results}
\label{sec:results}

We now extend the next generation matrix method for computing the basic reproduction number of an infectious disease to computing the boundary reproduction number in a chemical reaction network. Our methodology and notation closely follow that of \cite{VANDENDRIESSCHE2002} and formalize the work of \cite{Avram2025StabilityIR} in the language of chemical reaction networks.

\subsection{Boundary Reproduction Number}

Consider a chemical reaction network $(\mathcal{S},\mathcal{R})$ which has a critical siphon $\mathcal{X} \subset \mathcal{S}$. We reorder the species set so that the first $m$ species are elements of $\mathcal{X}$ (i.e. $X_i \in \mathcal{X}$ for $i=1, \ldots, m$) and the remaining $n-m$ are elements of $\mathcal{Y}$ (i.e. $Y_{i}=X_{i+m} \in \mathcal{Y}$ for $i = 1, \ldots, n-m$). Let $\tilde{\mathbf{x}} \in \mathbb{R}^m_{\geq 0}$ denote the concentration vector for species in the siphon $\mathcal{X}$ and $\tilde{\mathbf{y}} \in \mathbb{R}^{n-m}_{\geq 0}$ denote the concentration vector for species in the antisiphon $\mathcal{Y}$, so that the state vector can be decomposed as $\mathbf{x} = (\tilde{\mathbf{x}},\tilde{\mathbf{y}}) \in \mathbb{R}_{\geq 0}^n$. We rewrite the system of differential equations \eqref{de} as
\begin{equation}
    \label{de2}
    \left\{ \; \; \;
    \begin{aligned}
    \frac{d\tilde{\mathbf{x}}}{dt} & = \mathbf{f}(\tilde{\mathbf{x}},\tilde{\mathbf{y}}) \\
    \frac{d\tilde{\mathbf{y}}}{dt} & = \mathbf{g}(\tilde{\mathbf{x}},\tilde{\mathbf{y}}).
    \end{aligned}
    \right.
\end{equation}
The decomposition \eqref{de2} is analogous to the decomposition in \cite{VANDENDRIESSCHE2002}, where the species in $\tilde{\mathbf{x}}$ are interpreted as infected and the species in $\tilde{\mathbf{y}}$ are noninfected. Note that, since $\mathcal{X}$ is a siphon, we have that $\mathbf{f}(\mathbf{0},\tilde{\mathbf{y}}) = \mathbf{0}$ for every $(\mathbf{0},\tilde{\mathbf{y}}) \in B_{\mathcal{X}}$, so that $B_{\mathcal{X}}$ is forward invariant. This is analogous to the epidemiological property that no new infections are generated without current infections.

The decomposition $\mathbf{x} = (\tilde{\mathbf{x}},\tilde{\mathbf{y}}) \in \mathbb{R}_{\geq 0}^n$ allows us to rewrite the conservation equation \eqref{conservation} as
\begin{equation}
\label{conservation2}
W \mathbf{x} = [W_{\mathcal{X}} \; W_{\mathcal{Y}}] \left[ \begin{array}{c} \tilde{\mathbf{x}} \\ \tilde{\mathbf{y}} \end{array} \right] = W_{\mathcal{X}} \tilde{\mathbf{x}} + W_{\mathcal{Y}} \tilde{\mathbf{y}} = \mathbf{\Lambda}.
\end{equation}
Notably, at any $\mathcal{X}$-free boundary steady state $\mathbf{x}^* = (\mathbf{0},\tilde{\mathbf{y}}^*) \in E_{\mathcal{X}}$, this reduces to $W_{\mathcal{Y}} \tilde{\mathbf{y}}^* = \mathbf{\Lambda}$, so that, if $W_{\mathcal{Y}}$ is invertible, we have the unique solution $\tilde{\mathbf{y}}^* = W_{\mathcal{Y}}^{-1} \mathbf{\Lambda}$. Also note that the values of $\mathbf{\Lambda}$ uniquely determine the stoichiometric compatibility class $\mathsf{C}_{\mathbf{x}_0}$.

We now decompose the vector $\mathbf{f}(\tilde{\mathbf{x}},\tilde{\mathbf{y}})$ as in \cite{VANDENDRIESSCHE2002,van2008further}. Let $\mathcal{F}(\tilde{\mathbf{x}},\tilde{\mathbf{y}}) \in \mathbb{R}^{m}_{\geq 0}$ be a vector that contains terms in $\mathbf{f}(\tilde{\mathbf{x}},\tilde{\mathbf{y}})$ which are positive. These terms correspond to terms which ``produce'' species in $\mathcal{X}$, analogously with producing new infections in disease-spread models. We define $-\mathcal{V}(\tilde{\mathbf{x}},\tilde{\mathbf{y}})$ to contain the remaining terms in $\mathbf{f}(\tilde{\mathbf{x}},\tilde{\mathbf{y}})$. Note that $\mathcal{F}(\tilde{\mathbf{x}},\tilde{\mathbf{y}})$ need not contain all of the positive terms in $\mathbf{f}(\tilde{\mathbf{x}},\tilde{\mathbf{y}})$, and the terms in $\mathcal{V}(\tilde{\mathbf{x}},\tilde{\mathbf{y}})$ may be positive or negative. This decomposition allows the first equation in \eqref{de2} to be rewritten as
\begin{equation}
    \label{de3}
    \frac{d\tilde{\mathbf{x}}}{dt} = \mathcal{F}(\tilde{\mathbf{x}},\tilde{\mathbf{y}}) - \mathcal{V}(\tilde{\mathbf{x}},\tilde{\mathbf{y}}).
\end{equation}

Given a boundary steady state $\mathbf{x}^* = (\mathbf{0},\tilde{\mathbf{y}}^*) \in E_{\mathcal{X}}$, we define the Jacobians
\begin{equation}
    \label{FV}
F = \frac{\partial \mathcal{F}}{\partial \tilde{\mathbf{x}}}(\mathbf{0},\tilde{\mathbf{y}}^*) \; \mbox{ and } \; V = \frac{\partial \mathcal{V}}{\partial \tilde{\mathbf{x}}} (\mathbf{0},\tilde{\mathbf{y}}^*) .
\end{equation}
We note that the choice of vectors $\mathcal{F}$ and $\mathcal{V}$ for splitting $\mathbf{f}(\tilde{\mathbf{x}},\tilde{\mathbf{y}})$ in \eqref{de3}, and consequently the matrices $F$ and $V$ for splitting the Jacobian of $\mathbf{f}(\tilde{\mathbf{x}},\tilde{\mathbf{y}})$ in \eqref{FV}, are not necessarily unique (see \cite{VANDENDRIESSCHE2002}). Heuristics for choosing $\mathcal{F}$ and $\mathcal{V}$ are provided in Section \ref{sec:heuristic}.

We introduce the following matrix classifications.

\begin{definition}
\label{matrices1}
Consider a matrix $A \in \mathbb{R}^{n \times n}$. Then:
\begin{enumerate}
\item $A$ is a \textbf{nonnegative matrix}, denoted $A \geq 0$, if $a_{ij} \geq 0$ for all $i, j =1, \ldots, n$ and $A \not= 0$.
\item $A$ is a \textbf{$Z$-matrix} if $a_{ij} \leq 0$ for all $i \not= j$.
\end{enumerate}
\end{definition}
\noindent We also refer to nonnegative vectors $\mathbf{v} \geq 0$ as those which satisfy $v_i \geq 0$ for all $i = 1, \ldots, n$ and $\mathbf{v} \not= \mathbf{0}.$

We now introduce the following concept, which is the primary object of the paper and extends the basic reproduction number from mathematical epidemiology to mathematical biochemistry \cite{Avram2025StabilityIR}.

\begin{definition}
\label{def:bpn}
Let $\mathcal{X} \subseteq \mathcal{S}$ be a critical siphon of a chemical reaction network $(\mathcal{S},\mathcal{R})$ with corresponding chemical reaction system \eqref{de}. Suppose $\mathbf{x}^* \in E_{\mathcal{X}}$ is an $\mathcal{X}$-free boundary steady state of \eqref{de} and let $\mathbf{f}(\tilde{\mathbf{x}},\tilde{\mathbf{y}})$ be split for some choice of $\mathcal{F}$ and $\mathcal{V}$ as in \eqref{de3} and $F$ and $V$ be defined as in \eqref{FV}. Suppose that $F$ is nonnegative and $V$ is a $Z$-matrix with a nonnegative inverse (i.e. $F \geq 0$ and $V^{-1} \geq 0$). Then the \textbf{boundary reproduction number of the steady state} $\mathbf{x}^*$ for the splitting $\mathcal{F}$ and $\mathcal{V}$ is $\mathscr{R}_{\mathbf{x}^*} = \rho(FV^{-1})$.
\end{definition}

\noindent We interpret the boundary reproduction number analogously to the basic reproduction number $\mathscr{R}_0$ in mathematical epidemiology. Specifically, a boundary reproduction number less than one implies that a small infusion of species in $\mathcal{X}$ will tend to diminish in the system, while a boundary reproduction number greater than one implies that a small infusion of species in $\mathcal{X}$ will grow. %For $\mathscr{R}_{\mathbf{x}^*}$, this interpretation holds locally near the $\mathcal{X}$-free boundary steady state $\mathbf{x}^* \in E_{\mathcal{X}}$ only while, for $\mathscr{R}_{\mathbf{x}^*}$, this interpretation holds for all steady states on a the $\mathcal{X}$-free boundary face $B_{\mathcal{X}}$ relative to the steady state's stoichiometric compatibility class.

%In general, a boundary steady state $\mathbf{x}^* \in E_{\mathcal{X}}$ depends upon the rate parameters of \eqref{de} (e.g. $k_1$, $k_2, \ldots$). Consequently, the boundary reproduction number of a steady state ($\mathscr{R}_{\mathbf{x}^*}$) generally depends on the rate parameters. For deadlocks, however, the entire boundary face $B_{\mathcal{X}}$ consists of steady states (i.e. $B_{\mathcal{X}} = E_{\mathcal{X}}$). Consequently, in this case we can compute a single persistence value for the entire boundary ($\mathscr{R}_{\mathbf{x}^*}$) rather than at isolated steady states $\mathbf{x}^*$. Generally, this value depends on not only the rate parameters but the conservation constants (e.g. $\Lambda_1$, $\Lambda_2, \ldots$) as well.

The decomposition of a matrix $A = M - N$ where $M \geq 0$ and $N^{-1} \geq 0$ is not unique to the next generation method in mathematical epidemiology \cite{VANDENDRIESSCHE2002,van2008further}. This technique of splitting has also been used in iterative methods (e.g., see \cite{varga1962iterative,rose1984convergent,berman1994nonnegative}) and Markov chain theory (e.g., see \cite{seneta2006non}), among other areas.

\subsection{Selection of $\mathcal{F}$ and $\mathcal{V}$}
\label{sec:heuristic}

In general, there are multiple options for sets of positive terms from \eqref{de2} to include in $\mathcal{F}$, and not all of them lead to a $V$ with a nonnegative inverse ($V^{-1} \geq 0$). Selecting the maximal set of positive terms for $\mathcal{F}$ was the heuristic adopted in \cite{avram2023algorithmic}. This approach increases the likelihood that $V^{-1} \geq 0$; however, it can also increase the computational complexity of computing the reproduction number. Ideally, one seeks to determine the fewest positive terms to include in $\mathcal{F}$ while satisfying the requirements of Definition \ref{def:bpn}, so that the boundary reproduction number is valid and can be efficiently computed.

In this paper, we adopt the following heuristic for selecting the positive terms in \eqref{de2} to include in $\mathcal{F}$. The heuristic is based on the structure of interactions between siphon elements in the underlying interaction network and is justified in Sections \ref{sec:mainresult}, \ref{sec:matrixresult}, and \ref{sec:networkresult} (in particular, Theorem \ref{thm:3}). We also show in our examples how this heuristic applies to models of infectious disease spread.
\begin{enumerate}
\item[\textbf{(H1)}] Identify and include in $\mathcal{F}$ any positive term corresponding to an autocatalytic reaction on a given species. A reaction $R_j$ is considered autocatalytic in $X_i$ if $\alpha_{ij} > 0$, $\beta_{ij} > 0$, and $\beta_{ij} - \alpha_{ij} > 0$ (i.e., $X_i$ is a reactant and product in the same reaction and increases as a result of the reaction).
\item[\textbf{(H2)}] Identify and potentially include in $\mathcal{F}$ positive terms on the right-hand side of \eqref{de2} that result from reactions disassociative in $\mathcal{X}$. A reaction is disassociative in $\mathcal{X}$ if there are multiple species from the siphon in the product (i.e., $\beta_{ik} > 0$ and $\beta_{jk} > 0$ for $X_i, X_j \in \mathcal{X}$).
\item[\textbf{(H3)}] Identify and potentially include in $\mathcal{F}$ any remaining positive terms on the right-hand side of \eqref{de2}.
\end{enumerate}

We construct $\mathcal{F}$ by starting with $\mathcal{F} = \mathbf{0}$, iteratively adding terms to $\mathcal{F}$, and checking whether $V^{-1} \geq 0$ is attained. We add terms to $\mathcal{F}$ in the order of heuristic \textbf{(H1)} (include all terms), then \textbf{(H2)} (include terms as necessary), then \textbf{(H3)} (include terms as necessary). Theorem \ref{thm:3} provides further insight into how to efficiently select these terms by giving sufficient conditions under which a choice of terms succeeds.

\begin{example}
\label{ex:EnvZOmpR1}
Reconsider the network in \eqref{EnvZOmpr-simplified} from Section \ref{sec:introduction}. We have that $\mathcal{X} = \{ X, X_t, Y, X_pY, X_tY_p \}$ is a siphon with corresponding antisiphon $\mathcal{Y} = \{ X_p, Y_p \}$. For clarity, we relabel the species set $\mathcal{S}$ according to $X_1 = X, X_2 = X_t, X_3 = Y, X_4 = X_pY, X_5 = X_tY_p, Y_1 = X_p$, and $Y_2 = Y_p$, which gives:
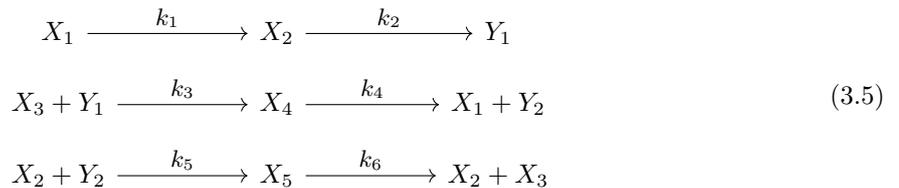
\begin{equation}
\label{EnvZOmpr-simplified2}
\begin{tikzcd}
X_1 \arrow[rr,"k_1"] & & X_2 \arrow[rr,"k_2"] & & Y_1 \\[-0.1in]
X_3 + Y_1 \arrow[rr,"k_3"] & & X_4 \arrow[rr,"k_4"] & & X_1 + Y_2 \\[-0.1in]
X_2 + Y_2 \arrow[rr,"k_5"] & & X_5 \arrow[rr,"k_6"] & & X_2 + X_3
\end{tikzcd}
\end{equation}

The system \eqref{EnvZOmpr-simplified2} has the conservation laws $X_1+X_2+X_4+X_5+Y_1 = X_{tot}$ and $X_3+X_4+X_5+Y_2 = Y_{tot}$ with conservation constants $X_{tot} > 0$ and $Y_{tot} > 0$, so that $\mathcal{X}$ is critical. The associated mass-action system, partitioned as in \eqref{de2}, is:
\begin{equation}
\label{EnvZOmpr-simplified2-de}
\left\{ \; \; \begin{aligned}
x_1' & = -k_1x_1+k_4x_4 \\
x_2' & = k_1x_1-k_2x_2 - k_5 x_2 y_2 + k_6 x_5 \\
x_3' & = -k_3 x_3 y_1+k_6x_5 \\
x_4' & = k_3 x_3 y_1 - k_4 x_4 \\
x_5' & = k_5 x_2 y_2 - k_6 x_5
\end{aligned}\right. \; \; 
\left\{ \; \; \begin{aligned}
y_1' & = k_2x_2-k_3x_3 y_1 \\
y_2' & = k_4x_4-k_5x_2 y_2
\end{aligned}
\right.
\end{equation}
The conservation equation \eqref{conservation2} corresponding to \eqref{EnvZOmpr-simplified2} has
\[W_{\mathcal{X}} = \begin{bmatrix} 1 & 1 & 0 & 1 & 1 \\ 0 & 0 & 1 & 1 & 1 \end{bmatrix}, \quad W_{\mathcal{Y}} = \begin{bmatrix} 1 & 0 \\ 0 & 1 \end{bmatrix}, \quad \text{and} \quad \mathbf{\Lambda} = \begin{bmatrix} X_{tot} \\ Y_{tot} \end{bmatrix}.\]
We have that $|\mathcal{Y}| = n-s = 2$ and $W_{\mathcal{Y}}$ is invertible. Consequently, every $\mathcal{X}$-free boundary steady state has the form $\mathbf{x}^* = W_{\mathcal{Y}}^{-1}\mathbf{\Lambda} = (0,0,0,0,0,X_{tot},Y_{tot})$. Since the steady state $\mathbf{x}^*$ depends on the conservation constants $X_{tot}$ and $Y_{tot}$, $\mathscr{R}_{\mathbf{x}^*} = \rho(FV^{-1})$ may also depend on these constants.

We now consider which terms to include in the vector $\mathcal{F}$. The network \eqref{EnvZOmpr-simplified2} does not contain any autocatalytic reactions, so heuristic \textbf{(H1)} does not apply; however, it does contain the reaction $X_5 \stackrel{k_6}{\longrightarrow} X_2 + X_3$, which is disassociative in $\mathcal{X}$. Heuristic \textbf{(H2)} suggests choosing one of the $k_6 x_5$ terms from either the $x_2'$ or $x_3'$ equation in \eqref{EnvZOmpr-simplified2-de}. Choosing the $k_6 x_5$ term in the $x_2'$ equation gives:
\begin{equation}
\label{FV5}
\begin{aligned}
\mathcal{F} & = \langle 0,k_6 x_5,0,0,0 \rangle,\\
\mathcal{V} & = \langle k_1x_1-k_4x_4,-k_1x_1+k_2x_2+k_5 x_2 y_2,k_3 x_3 y_1-k_4x_5,-k_3 x_3 y_1+k_4 x_4,-k_5 x_2 y_2+k_6 x_5 \rangle.
\end{aligned}
\end{equation}

Evaluating the Jacobians of $\mathcal{F}$ and $\mathcal{V}$ with respect to the siphon elements $\mathcal{X} = \{ X_1, X_2, X_3, X_4, X_5 \}$ at $\mathbf{x}^* \in E_{\mathcal{X}}$ gives
\begin{equation}
\label{FV2}
F = \left[ \begin{array}{ccccc}
0&0&0&0&0 \\
0&0&0&0&k_6 \\
0&0&0&0&0 \\
0&0&0&0&0 \\
0&0&0&0&0
\end{array}\right]
\quad \text{and} \quad
V = \left[ \begin{array}{ccccc}
k_1 & 0 & 0 & -k_4 & 0 \\
-k_1 & k_5 Y_{tot} + k_2 & 0 & 0 & 0 \\
0 & 0 & k_3 X_{tot} & 0 & -k_6 \\
0 & 0 & -k_3 X_{tot} & k_4 & 0 \\
0 & -k_5 Y_{tot} & 0 & 0 & k_6
\end{array}
 \right]
\end{equation}
We have that $F \geq 0$ and $V$ is a $Z$-matrix. % The Jacobian of the $\mathbf{f}(\tilde{\mathbf{x}},\tilde{\mathbf{y}})$ given by the first five equations of \eqref{EnvZOmpr-simplified2-de} and evaluated at the boundary steady state $(\mathbf{x}^*)$ can be decomposed as $J(\mathbf{x}^*) = F - V$.
%\end{example}
%\begin{example}
%\label{ex:EnvZOmpR3}
%Reconsider the network \eqref{EnvZOmpr-simplified2} introduced in Section \ref{sec:introduction} and analyzed in Example \ref{ex:EnvZOmpR1} and \ref{ex:EnvZOmpR2}. 
We compute
\begin{equation}
    \label{Vinv}
\displaystyle{V^{-1} = \left[ 
\begin{array}{ccccc} 
\displaystyle \frac{k_2 + k_5 Y_{tot}}{k_1 k_2} & \displaystyle \frac{k_5 Y_{tot}}{k_1 k_2} & \displaystyle \frac{k_2 + k_5 Y_{tot}}{k_1 k_2} & \displaystyle \frac{k_2 + k_5 Y_{tot}}{k_1 k_2} & \displaystyle \frac{k_2 + k_5 Y_{tot}}{k_1 k_2} \\[0.1in]
\displaystyle \frac{1}{k_2} & \displaystyle \frac{1}{k_2} & \displaystyle \frac{1}{k_2} & \displaystyle \frac{1}{k_2} & \displaystyle \frac{1}{k_2} \\[0.1in]
\displaystyle \frac{k_5 Y_{tot}}{k_2 k_3 X_{tot}} & \displaystyle \frac{k_5 Y_{tot}}{k_2 k_3 X_{tot}} & \displaystyle \frac{k_2 + k_5 Y_{tot}}{k_2 k_3 X_{tot}} & \displaystyle \frac{k_5 Y_{tot}}{k_2 k_3 X_{tot}} & \displaystyle \frac{k_2 + k_5 Y_{tot}}{k_2 k_3 X_{tot}} \\[0.1in]
\displaystyle \frac{k_5 Y_{tot}}{k_2 k_4} & \displaystyle \frac{k_5 Y_{tot}}{k_2 k_4} & \displaystyle \frac{k_2 + k_5 Y_{tot}}{k_2 k_4} & \displaystyle \frac{k_2 + k_5 Y_{tot}}{k_2 k_4} & \displaystyle \frac{k_2 + k_5 Y_{tot}}{k_2 k_4} \\[0.1in]
\displaystyle \frac{k_5 Y_{tot}}{k_2 k_6} & \displaystyle \frac{k_5 Y_{tot}}{k_2 k_6} & \displaystyle \frac{k_5 Y_{tot}}{k_2 k_6} & \displaystyle \frac{k_5 Y_{tot}}{k_2 k_6} & \displaystyle \frac{k_2 + k_5 Y_{tot}}{k_2 k_6}
\end{array}
\right]}
\end{equation}
so that $V^{-1} \geq 0$ for all rate parameter and conservation values. We therefore compute
\[FV^{-1} = \left[ \begin{array}{ccccc}
\displaystyle 0 & \displaystyle 0 & \displaystyle 0 & \displaystyle 0 & \displaystyle 0 \\[0.1in]
\displaystyle \frac{k_5 Y_{\text{tot}}}{k_2} & \displaystyle \frac{k_5 Y_{\text{tot}}}{k_2} & \displaystyle \frac{k_5 Y_{\text{tot}}}{k_2} & \displaystyle \frac{k_5 Y_{\text{tot}}}{k_2} & \displaystyle \frac{k_5 Y_{\text{tot}} + k_2}{k_2} \\[0.1in]
\displaystyle 0 & \displaystyle 0 & \displaystyle 0 & \displaystyle 0 & \displaystyle 0 \\[0.1in]
\displaystyle 0 & \displaystyle 0 & \displaystyle 0 & \displaystyle 0 & \displaystyle 0 \\[0.1in]
\displaystyle 0 & \displaystyle 0 & \displaystyle 0 & \displaystyle 0 & \displaystyle 0
\end{array}
 \right].\]
The boundary reproduction number of $\mathbf{x}^*$ for this choice of splitting $\mathcal{F}$ and $\mathcal{V}$ is therefore given by
\begin{equation}
    \label{bpn1}
\mathscr{R}_{\mathbf{x}^*} = \rho(FV^{-1}) = \displaystyle{\frac{k_5Y_{tot}}{k_2}}.
\end{equation}
%The conditions  $\mathscr{R}_{\mathbf{x}^*} < 1$ and $\mathscr{R}_{\mathbf{x}^*} > 1$ can be solved for $Y_{tot}$ to show that the boundary steady state $(x^*_1,x^*_2,x^*_3,x^*_4,x^*_5,y^*_1,y^*_2) = (0,0,0,0,0,X_{tot},Y_{tot})$ is locally asymptotically stable if
%\[Y_{tot} < \displaystyle{\frac{k_2}{k_5}}\]
%and unstable if
%\[Y_{tot} > \displaystyle{\frac{k_2}{k_5}}.\]
\end{example}

\subsection{Main Result}
\label{sec:mainresult}

The following is the first main result of this paper. It is modified from Theorem 2 of \cite{VANDENDRIESSCHE2002} and Theorem 1 of \cite{van2008further} for chemical reaction networks rather than models of infectious disease spread. The proof is given in Appendix \ref{app:a} and is modeled after \cite{van2008further}.

\begin{theorem}
\label{thm:main}
Let $\mathcal{X} \subseteq \mathcal{S}$ be a critical siphon of a chemical reaction network $(\mathcal{S},\mathcal{R})$ with corresponding chemical reaction system \eqref{de}. Suppose $\mathbf{x}^* \in E_{\mathcal{X}}$ is an $\mathcal{X}$-free boundary steady state of \eqref{de}, and let $\mathbf{f}(\tilde{\mathbf{x}},\tilde{\mathbf{y}})$ be split for some choice of $\mathcal{F}$ and $\mathcal{V}$ as in \eqref{de3}, and $F$ and $V$ be defined as in \eqref{FV}.
\begin{enumerate}
\item[(a)] Suppose the network has no conservation laws and the eigenvalues of the Jacobian of $\mathbf{g}(\tilde{\mathbf{x}},\tilde{\mathbf{y}})$ with respect to $\tilde{\mathbf{y}}$ evaluated at $\mathbf{x}^*$ have negative real parts. Then $\mathbf{x}^*$ is locally asymptotically stable if $\mathscr{R}_{\mathbf{x}^*} < 1$ and unstable if $\mathscr{R}_{\mathbf{x}^*} > 1$.
\item[(b)] Suppose the number of elements in the antisiphon $\mathcal{Y}$ equals the number of linearly independent conservation laws (i.e., $|\mathcal{Y}| = n - s$) and that $W_{\mathcal{Y}}$ is invertible. Then the steady state $\mathbf{x}^* = (0, W_{\mathcal{Y}}^{-1}\mathbf{\Lambda})$ is the unique steady state within its stoichiometric compatibility class, and it is locally asymptotically stable with respect to $\mathsf{C}_{\mathbf{x}_0}$ if $\mathscr{R}_{\mathbf{x}^*} < 1$ and unstable with respect to $\mathsf{C}_{\mathbf{x}_0}$ if $\mathscr{R}_{\mathbf{x}^*} > 1$.
\end{enumerate}
\end{theorem}

\noindent Theorem \ref{thm:main} includes, but is more general than, Theorem 1 of \cite{van2008further}. In particular, conditions (A1)–(A3) of Theorem 1 of \cite{van2008further} are satisfied by kinetic condition \textbf{(A2)} and the definition of the siphon $\mathcal{X}$ (Definition \ref{def:siphon}), but condition (A4) is violated by disassociation reactions, i.e., reactions like $X_1 \to X_2 + X_3$ where $X_1, X_2, X_3 \in \mathcal{X}$. It is also worth noting that conservation laws, which are common in mathematical biochemistry, were not considered in \cite{VANDENDRIESSCHE2002} or \cite{van2008further}. When the number of elements of $\mathcal{Y}$ and the number of linearly independent conservation laws coincide, the boundary steady state in each stoichiometric compatibility class depends uniquely on the conservation constants. In general, we may have $\mathscr{R}_{\mathbf{x}^*} > 1$ for some stoichiometric compatibility classes but $\mathscr{R}_{\mathbf{x}^*} < 1$ for others.

\begin{example}
\label{ex:EnvZOmpR2}
Reconsider the network in \eqref{EnvZOmpr-simplified} from Section \ref{sec:introduction} and analyzed in Example \ref{ex:EnvZOmpR1}. We have $|\mathcal{Y}| = n - 2 = 2$ and $W_{\mathcal{Y}}$ is invertible. Furthermore, for every $\mathbf{x}^* \in E_\mathcal{X}$, we have $F \geq 0$, $V$ is a $Z$-matrix, and $V^{-1} \geq 0$ \eqref{Vinv}, so we may define $\mathscr{R}_{\mathbf{x}^*}$ as in \eqref{bpn1} by Definition \ref{def:bpn}. 

By Theorem \ref{thm:main}, the boundary steady state $\mathbf{x}^* \in E_{\mathcal{X}}$ is locally asymptotically stable with respect to its stoichiometric compatibility class if
\[
\mathscr{R}_{\mathbf{x}^*} = \rho(FV^{-1}) = \frac{k_5Y_{tot}}{k_2} < 1 \Longrightarrow Y_{tot} < \frac{k_2}{k_5}
\]
and unstable with respect to its stoichiometric compatibility class if
\[
\mathscr{R}_{\mathbf{x}^*} = \rho(FV^{-1}) = \frac{k_5Y_{tot}}{k_2} > 1 \Longrightarrow Y_{tot} > \frac{k_2}{k_5}.
\]

We highlight that establishing the stability threshold for the boundary steady state is challenging by direct methods. One alternative to applying Theorem \ref{thm:main} is to compute the eigenvalues of the Jacobian of \eqref{EnvZOmpr-simplified2-de} evaluated at the boundary steady state, and then establish conditions under which their real parts are negative. The Jacobian of the system restricted to the siphon $\mathcal{X}$ is given by
\begin{equation}
\label{jacobian}
J(\mathbf{x}^*) = \left[ \begin{array}{ccccc} 
-k_1 & 0 & 0 & k_4 & 0 \\
k_1 & - k_5 Y_{tot}-k_2  & 0 & 0 & k_6 \\
0 & 0 & -k_3 X_{tot} & 0 & k_6 \\
0 & 0 & k_3 X_{tot} & -k_4 & 0 \\
0 & k_5 Y_{tot} & 0 & 0 & -k_6 
\end{array}
\right].
\end{equation}
This matrix is irreducible and computing the eigenvalues of \eqref{jacobian} requires solving a non-factorable fifth-order polynomial over eight symbolic parameters. This can only be done in exceptional cases, which are not expected in general. Alternatively, one may compute the Routh–Hurwitz table to determine conditions under which the real parts of the eigenvalues are negative. For the matrix \eqref{jacobian}, however, the Routh–Hurwitz table would require over a page to display, even after simplification on a computer, and is consequently too lengthy to be of practical use. The boundary reproduction number presented in Theorem \ref{thm:main} thus represents a significant improvement in computational efficiency.
\end{example}

\subsection{Matrix Structure Result}
\label{sec:matrixresult}

In this section, we further consider the structure of the matrix $V$. We introduce the following matrix structures.

\begin{definition}
\label{matrices}
A matrix $A \in \mathbb{R}^{n \times n}$ is \textbf{block triangularizable} if there exists a permutation matrix $P$ such that
\begin{equation}
\label{blocktriangular}
PAP^T = \left[ \begin{array}{cccc}
A_1 & J_{12} & \cdots & J_{1n} \\
0 & A_2 & \cdots & J_{2n} \\
\vdots & \vdots & \ddots & \vdots \\
0 & 0 & \cdots & A_n
\end{array}\right]
\end{equation}
where $A_i$ are irreducible square matrices, $J_{ij}$ are rectangular matrices of appropriate size, and $0$ are rectangular matrices of all zeroes.
\end{definition}

We now present a result on how the matrix structures of $F$ and $V$ can guarantee that the conditions of Theorem \ref{thm:main} hold. The proof is given in Appendix \ref{app:b}.
\begin{theorem}
\label{thm:2}
Let $\mathcal{X} \subseteq \mathcal{S}$ be a critical siphon of a chemical reaction network $(\mathcal{S},\mathcal{R})$ with corresponding chemical reaction system \eqref{de}. Define $F$ and $V$ as in \eqref{FV}, and suppose that $F \geq 0$ and that the matrix $V$ is block triangularizable with block triangular form \eqref{blocktriangular}, where $J_{ij} \leq 0$ and, for each $A_i$, $i = 1, \ldots, n$, $A_i$ is a $Z$-matrix and $\mathbbm{1}^T A_i \geq 0$. Then $V$ has a nonnegative inverse (i.e., $V^{-1} \geq 0$).
\end{theorem}

\begin{example}
\label{ex:EnvZOmpR4}
Reconsider the network \eqref{EnvZOmpr-simplified2} introduced in Section \ref{sec:introduction} and analyzed in Example \ref{ex:EnvZOmpR1}. The matrix $V$ in \eqref{FV2} is irreducible and therefore already in block triangular form. Since $V$ is a $Z$-matrix and $\mathbbm{1}^T V = (0, k_2, 0, 0, 0) \geq 0$, Theorem \ref{thm:2} is satisfied. Consequently, $V^{-1} \geq 0$, and the boundary reproduction number $\mathscr{R}_{\mathbf{x}^*} = \rho(FV^{-1})$ found in \eqref{bpn1} is valid by Theorem \ref{thm:main}.
\end{example}

\begin{example}
\label{ex:EnvZOmpR6}
Reconsider the network \eqref{EnvZOmpr-simplified2} introduced in Section \ref{sec:introduction} but instead of choosing the $k_6x_5$ term in the $x_2'$ equation we choose $k_6x_5$ in the $x_3'$ equation in \eqref{EnvZOmpr-simplified2-de}. This is consistent with heuristic \textbf{(H2)} and gives
\begin{equation}
\label{FV4}
\begin{aligned}
\mathcal{F} & = \langle 0,0,k_6x_5,0,0 \rangle,\\
\mathcal{V} & = \langle k_1x_1-k_4x_4,-k_6 x_5-k_1x_1+k_2x_2+k_5 x_2 y_2,k_3 x_3 y_1,-k_3 x_3 y_1+k_4 x_4,-k_5 x_2 y_2+k_6 x_5\rangle.
\end{aligned}
\end{equation}
This produces
\begin{equation}
\label{FV3}
F =\begin{bmatrix}
0 & 0 & 0 & 0 & 0 \\
0 & 0 & 0 & 0 & 0 \\
0 & 0 & 0 & 0 & k_6 \\
0 & 0 & 0 & 0 & 0 \\
0 & 0 & 0 & 0 & 0
\end{bmatrix}
\text{ and }
V = \begin{bmatrix}
k_1 & 0 & 0 & -k_4 & 0 \\
-k_1 & k_5 Y_{\text{tot}} + k_2  & 0 & 0 & -k_6 \\
0 & 0 & k_3 X_{\text{tot}} & 0 & 0 \\
0 & 0 & -k_3 X_{\text{tot}} & k_4 & 0 \\
0 & -k_5 Y_{\text{tot}} & 0 & 0 & k_6 \\
\end{bmatrix}
\end{equation}
This matrix $V$ in \eqref{FV3} is reducible. Permuting the matrix $V$ from \eqref{FV3} in the order $\{X_2, X_5, X_1, X_4, X_3\}$ gives the block triangular form
%\begin{equation}
%\label{P}
\[P V P^T = \begin{bmatrix}
k_5Y_{tot}  + k_2 & -k_6 & -k_1 & 0 & 0 \\
- k_5Y_{tot}  & k_6 & 0 & 0 & 0 \\
0 & 0 & k_1 & -k_4 & 0 \\
0 & 0 & 0 & k_4 & -k_3 X_{tot} \\
0 & 0 & 0 & 0 & k_3 X_{tot}
\end{bmatrix}\]
%\end{equation}
with blocks given by
\[A_1 = \left[ \begin{array}{cc} k_5Y_{tot} +k_2&-k_6 \\ -k_5Y_{tot} &k_6 \end{array} \right], A_2 = [ \;k_1 \; ], A_3 = [ \; k_4 \; ], \text{ and } A_4 = [ \; k_3 X_{tot} \; ].\]
We have that $\mathbbm{1}^T A_1 = (k_2,0) \geq 0$, $\mathbbm{1}^T A_2 = k_1 \geq 0$, $\mathbbm{1}^T A_3 = k_4 > 0$ and $\mathbbm{1}^T A_4 = k_3 X_{tot} > 0$. It follows from Theorem \ref{thm:2} that $V^{-1} \geq 0$ so that the boundary reproduction number produced by these matrices satisfies the conditions of Definition \ref{def:bpn} and Theorem \ref{thm:main}. %Note that for this example we have $\mathbbm{1}^T V = (0,k_2,k_3X_{tot}, 0, -k_6)$ does not satisfy $\mathbbm{1}^T V \geq 0$ so that the block structure \eqref{P} is necessary to guarantee $V^{-1} \geq 0$.

In this case, the choices $\mathcal{F}$ and $\mathcal{V}$ in \eqref{FV4} generate the same boundary reproduction number $\mathscr{R}_{\mathbf{x}^*}$ \eqref{bpn1} as the $\mathcal{F}$ and $\mathcal{V}$ in \eqref{FV5}. In general, two valid choices of $\mathcal{F}$ and $\mathcal{V}$ may produce different boundary reproduction numbers (see Examples \ref{ex:vectorborne}); however, in all cases the same threshold for stability and instability must be attained by Theorem \ref{thm:main}.
\end{example}
%Consequently, we have that the boundary reproduction number is not unique; rather, it is a function of the choice of splitting of the terms in $\mathcal{F}$ and $\mathcal{V}$. Multiple boundary reproduction numbers can validly produce the same stability thresholds since Theorem \ref{thm:main} and Theorem \ref{thm:2} apply.

%{\color{red}Comment on the selection of $\mathcal{F}$...}

\subsection{Network Structure Result}
\label{sec:networkresult}

In this section, we extend the matrix structure results of Section \ref{sec:matrixresult} to a particular network structure.

Consider a weighted directed graph $G = (V,E)$, where $V = \{ V_1, \ldots, V_n \}$ is a set of vertices and $E \subseteq V \times V$ is a set of directed edges. We represent edges $E_k = (V_i, V_j)$ as $V_i \to V_j$ and allow self-edges (i.e., edges of the form $(V_i, V_i) \in E$). A \emph{connected component} of $G$ is a maximal set of vertices connected by undirected paths. A \emph{strongly connected component} is a maximal set of vertices connected by directed paths. A \emph{terminal strongly connected component} is a strongly connected component with outward edges.

We now introduce the following network structure related to chemical reaction networks and siphons.

\begin{definition}
\label{def:xreduced}
Let $\mathcal{X} \subseteq \mathcal{S}$ be a critical siphon of a chemical reaction network $(\mathcal{S},\mathcal{R})$ with corresponding chemical reaction system \eqref{de}. The $\mathcal{X}$-reduced network is the graph $G_{\mathcal{X}} = (V_{\mathcal{X}},E_{\mathcal{X}})$ where the vertex set is $V_{\mathcal{X}} = \mathcal{X} \cup \{\emptyset\}$ and the edge set $E_{\mathcal{X}}$ is defined as follows:
\begin{enumerate}
    \item $(X_i, X_j) \in E_{\mathcal{X}}$ if there is a reaction $R_k \in \mathcal{R}$ such that $X_i$ is a reactant and $X_j$ is a product (i.e., $\alpha_{ik} > 0$ and $\beta_{jk} > 0$);
    \item $(X_i, X_i) \in E_{\mathcal{X}}$ if there is a reaction $R_k \in \mathcal{R}$ such that $X_i \in \mathcal{X}$ is both a reactant and a product and the amount of $X_i$ increases (i.e., $\alpha_{ik} > 0$, $\beta_{ik} > 0$, and $\beta_{ik}-\alpha_{ik} > 0$); and
    \item $(X_i, \emptyset) \in E_{\mathcal{X}}$ if there is a reaction $R_k \in \mathcal{R}$ such that $X_i$ is a reactant and no $X_j \in \mathcal{X}$ is a product (i.e., $\alpha_{ik} > 0$ and $\beta_{jk} = 0$ for all $X_j \in \mathcal{X}$) or a reaction $R_k \in \mathcal{R}$ such that $X_i \in \mathcal{X}$ is both a reactant and a product and the amount of $X_i$ decreases (i.e., $\alpha_{ik} >0$, $\beta_{ik} > 0$, and $\beta_{ik}-\alpha_{ik} < 0$).
\end{enumerate}
\end{definition}

To each edge $(X_i,X_j) \in E_{\mathcal{X}}$, $(X_i, X_i) \in E_{\mathcal{X}}$, or $(X_i, \emptyset) \in E_{\mathcal{X}}$, we associate the sum of the reaction rates $R_j(\mathbf{x})$ of the original reactions mapped to the edge by Definition \ref{def:xreduced}. Note that a single reaction rate $R_j(\mathbf{x})$ may be associated with multiple edges in $G_{\mathcal{X}}$. For example, the disassociation reaction $X_1 \stackrel{R(\mathbf{x})}{\longrightarrow} X_2 + X_3$ with $X_1, X_2, X_3 \in \mathcal{X}$ corresponds to weighted edges $X_1 \stackrel{R(\mathbf{x})}{\longrightarrow} X_2 \quad \text{and} \quad X_1 \stackrel{R(\mathbf{x})}{\longrightarrow} X_3$ in $G_{\mathcal{X}}$. Similarly, a single edge in $G_{\mathcal{X}}$ may be associated with multiple reactions from the original network. For example, the reactions $X_1 \stackrel{R_{1}(\mathbf{x})}{\longrightarrow} X_2 + Y_1 \quad \text{and} \quad X_1 \stackrel{R_{2}(\mathbf{x})}{\longrightarrow} X_2 + Y_2$, where $X_1, X_2 \in \mathcal{X}$ and $Y_1, Y_2 \in \mathcal{Y}$, are condensed as $X_1 \stackrel{R_{1}(\mathbf{x}) + R_{2}(\mathbf{x})}{\longrightarrow} X_2.$

\begin{definition}
\label{def:splitting}
Consider the $\mathcal{X}$-reduced network $G_{\mathcal{X}} = (V_{\mathcal{X}}, E_{\mathcal{X}})$ of a chemical reaction network $(\mathcal{S},\mathcal{R})$ with a critical siphon $\mathcal{X}$. A \emph{splitting} of $G_{\mathcal{X}}$ is any pair of networks $G^{\mathcal{F}}_{\mathcal{X}} = (V_{\mathcal{X}}, E^{\mathcal{F}}_{\mathcal{X}})$ and $G^{\mathcal{V}}_{\mathcal{X}} = (V_{\mathcal{X}}, E^{\mathcal{V}}_{\mathcal{X}})$ such that $E^{\mathcal{F}}_{\mathcal{X}} \cup E^{\mathcal{V}}_{\mathcal{X}} = E_{\mathcal{X}}$ and $E^{\mathcal{F}}_{\mathcal{X}} \cap E^{\mathcal{V}}_{\mathcal{X}} = \emptyset$.
\end{definition}

We now give sufficient conditions on the $\mathcal{X}$-reduced network to guarantee the validity of the boundary reproduction number. To this end, we require that $F$ and $V$ in \eqref{FV} capture the behavior of system \eqref{de} on the boundary face $B_{\mathcal{X}}$. Thus, we impose the following additional nondegeneracy assumption on the reaction rates $R_i(\mathbf{x})$, supplementing assumptions \textbf{(A1)}--\textbf{(A3)} from Section \ref{sec:crn}:

\begin{enumerate}
\item[\textbf{(A4)}] The dependence of the reaction rates on the reactant species does not vanish at boundary steady states $E_{\mathcal{X}}$ (i.e., if $\alpha_{ij} > 0$ and $\mathbf{x}^* = (\mathbf{0}, \tilde{\mathbf{y}}^*) \in E_{\mathcal{X}}$, then $\displaystyle \frac{\partial R_j}{\partial x_i}(\mathbf{x}^*) > 0$).
\end{enumerate}

Condition \textbf{(A4)} ensures that every reaction contributes to $F$ and $V$ as defined in \eqref{FV}. It holds for mass-action kinetics when $\alpha_{ij} = 1$, but fails when the stoichiometric coefficient is greater than one (i.e., $\alpha_{ij} > 1$). The following result on the $\mathcal{X}$-reduced network is proved in Appendix \ref{app:c}.

\begin{theorem}
\label{thm:3}
Let $\mathcal{X} \subseteq \mathcal{S}$ be a critical siphon of a chemical reaction network $(\mathcal{S},\mathcal{R})$ with corresponding chemical reaction system \eqref{de}. Suppose kinetic assumptions \textbf{(A1)}--\textbf{(A4)} hold. Further, suppose there exists a splitting $G_{\mathcal{X}} = G^{\mathcal{F}}_{\mathcal{X}} \cup G^{\mathcal{V}}_{\mathcal{X}},$ where $G^{\mathcal{F}}_{\mathcal{X}} = (V_{\mathcal{X}}, E^{\mathcal{F}}_{\mathcal{X}})$ and $G^{\mathcal{V}}_{\mathcal{X}} = (V_{\mathcal{X}}, E^{\mathcal{V}}_{\mathcal{X}})$ satisfying:
\begin{enumerate}
\item $E^{\mathcal{F}}_{\mathcal{X}}$ is nonempty;
\item the unique terminal strongly connected component of $G^{\mathcal{V}}_{\mathcal{X}}$ is $\{\emptyset\}$;
\item $E^{\mathcal{V}}_{\mathcal{X}}$ contains no self-loops (i.e., all self-loops, if any, belong to $E^{\mathcal{F}}_{\mathcal{X}}$); and
\item if $e_i, e_j \in E^{\mathcal{V}}_{\mathcal{X}}$ correspond to the same edge in the original network $(\mathcal{S},\mathcal{R})$, then:
\begin{enumerate}
\item $e_i$ and $e_j$ have the same source vertex; and
\item the target vertices of $e_i$ and $e_j$ belong to distinct strongly connected components.
\end{enumerate}
\end{enumerate}
Selecting $\mathcal{F}$ to contain terms from the right-hand side of \eqref{de2} corresponding to the targets of reactions in $E^{\mathcal{V}}_{\mathcal{X}}$ and $\mathcal{V}$ to contain the remaining terms, the matrix $V$ has a nonnegative inverse (i.e., $V^{-1} \geq 0$).
\end{theorem}

\begin{example}
\label{ex:EnvZOmpR5}
Reconsider the network \eqref{EnvZOmpr-simplified2} introduced in Section \ref{sec:introduction} and analyzed in Examples \ref{ex:EnvZOmpR1}, \ref{ex:EnvZOmpR2}, \ref{ex:EnvZOmpR4}, and \ref{ex:EnvZOmpR6}. Removing the species in $\mathcal{Y} = \{ Y_1, Y_2 \}$ from \eqref{EnvZOmpr-simplified2} gives
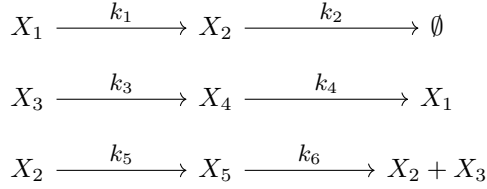
\[
\begin{tikzcd}
X_1 \arrow[rr,"k_1"] & & X_2 \arrow[rr,"k_2"] & & \emptyset \\[-0.1in]
X_3 \arrow[rr,"k_3"] & & X_4 \arrow[rr,"k_4"] & & X_1 \\[-0.1in]
X_2 \arrow[rr,"k_5"] & & X_5 \arrow[rr,"k_6"] & & X_2 + X_3
\end{tikzcd}
\]
where the rate constants uniquely determine the rate form from the mass-action system \eqref{EnvZOmpr-simplified2-de}. In accordance with steps 1–3 of Definition \ref{def:xreduced}, we split the reaction $X_5 \stackrel{k_6}{\longrightarrow} X_2 + X_3$ into $X_5 \stackrel{k_6}{\longrightarrow} X_2$ and $X_5 \stackrel{k_6}{\longrightarrow} X_3$ to obtain the following $\mathcal{X}$-reduced interaction network $G_{\mathcal{X}} = (V_{\mathcal{X}},E_{\mathcal{X}})$ and potential splittings:
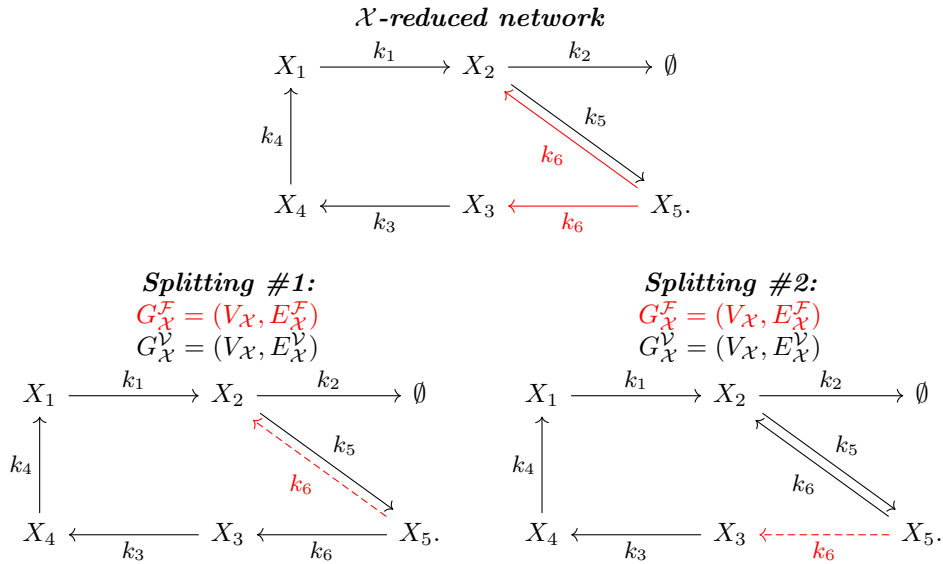
\[
\begin{array}{c}
\begin{array}{c}
\text{\textbf{$\mathcal{X}$-reduced network}}\\
\begin{tikzcd}
X_1 \arrow[rr,"k_1"] & & X_2 \arrow[rr,"k_2"] \arrow[ddrr,xshift=0.3ex,yshift=0.3ex,"k_5"] & & \emptyset \\
 & & & & \\
X_4 \arrow[uu,"k_4"] & & X_3 \arrow[ll,"k_3"] & & X_5. \arrow[lluu,xshift=-0.3ex,yshift=-0.3ex,red,"k_6"] \arrow[ll,red,"k_6"]
\end{tikzcd}
\end{array} \\ \\
\begin{array}{c}
\text{\textbf{Splitting \#1:}}\\
{\color{red}\text{$G^{\mathcal{F}}_{\mathcal{X}} = (V_{\mathcal{X}},E^{\mathcal{F}}_{\mathcal{X}})$}}\\
\text{$G^{\mathcal{V}}_{\mathcal{X}} = (V_{\mathcal{X}},E^{\mathcal{V}}_{\mathcal{X}})$}\\
\begin{tikzcd}
X_1 \arrow[rr,"k_1"] & & X_2 \arrow[rr,"k_2"] \arrow[ddrr,xshift=0.3ex,yshift=0.3ex,"k_5"] & & \emptyset \\
 & & & & \\
X_4 \arrow[uu,"k_4"] & & X_3 \arrow[ll,"k_3"] & & X_5. \arrow[lluu,xshift=-0.3ex,yshift=-0.3ex,red,dashed,"k_6"] \arrow[ll,"k_6"]
\end{tikzcd}
\end{array} \quad
\begin{array}{c}
\text{\textbf{Splitting \#2:}}\\
{\color{red}\text{$G^{\mathcal{F}}_{\mathcal{X}} = (V_{\mathcal{X}},E^{\mathcal{F}}_{\mathcal{X}})$}}\\
\text{$G^{\mathcal{V}}_{\mathcal{X}} = (V_{\mathcal{X}},E^{\mathcal{V}}_{\mathcal{X}})$}\\
\begin{tikzcd}
X_1 \arrow[rr,"k_1"] & & X_2 \arrow[rr,"k_2"] \arrow[ddrr,xshift=0.3ex,yshift=0.3ex,"k_5"] & & \emptyset \\
 & & & & \\
X_4 \arrow[uu,"k_4"] & & X_3 \arrow[ll,"k_3"] & & X_5.  \arrow[lluu,xshift=-0.3ex,yshift=-0.3ex,"k_6"] \arrow[ll,red,dashed,"k_6"]
\end{tikzcd}
\end{array}
\end{array}
\]

The splittings of the $X$-reduced network (top, center) into $G_{\mathcal{X}}^{\mathcal{F}}$ and $G_{\mathcal{X}}^{\mathcal{V}}$ correspond to the division of terms included in $\mathcal{F}$ and $\mathcal{V}$, respectively. This representation provides insight into how to split the network and why some splittings succeed while others do not. Crucially, we must remove one of the red edges from $E^{\mathcal{V}}_{\mathcal{X}}$ to avoid violating Condition 4(b) of Theorem \ref{thm:3}.

Splitting \#1 (bottom, left) corresponds to the $\mathcal{F}$ and $\mathcal{V}$ chosen in Example \ref{ex:EnvZOmpR4}. Specifically, we remove the edge $X_5 \stackrel{k_6}{\longrightarrow} X_2$ (dashed red) from $E^{\mathcal{V}}_{\mathcal{X}}$ and observe that the resulting network has no self-loops, with strongly connected components $\{ X_1, X_2, X_3, X_4, X_5 \}$ and $\{ \emptyset \}$, where the singleton $\{ \emptyset \}$ is terminal, and no strongly connected component contains multiple edges labeled with the same rate function. Consequently, Theorem \ref{thm:3} is satisfied, guaranteeing a matrix $V$ that satisfies the block triangularization assumption of Theorem \ref{thm:2}.

Splitting \#2 (bottom, right) corresponds to the $\mathcal{F}$ and $\mathcal{V}$ chosen in Example \ref{ex:EnvZOmpR6}. Specifically, we remove the edge $X_3 \stackrel{k_3}{\longrightarrow} X_4$ (dashed red) from $E^{\mathcal{V}}_{\mathcal{X}}$ and observe that the resulting network has no self-loops, with strongly connected components $\{ X_1 \}, \{ X_2, X_5 \}, \{ X_3 \}, \{ X_4 \}$, and $\{ \emptyset \}$, where the singleton $\{ \emptyset \}$ is terminal, and no strongly connected component contains multiple edges labeled with the same rate function. Consequently, Theorem \ref{thm:3} is satisfied, guaranteeing a matrix $V$ that satisfies the block triangularization assumption of Theorem \ref{thm:2}.
\end{example}

\section{Examples}
\label{sec:examples}

In this section, we present several examples illustrating the application of Theorems \ref{thm:main}, \ref{thm:2}, and \ref{thm:3}. In Section \ref{sec:infectiousexamples}, we show how the boundary reproduction number approach includes the basic reproduction number from mathematical epidemiology. In Section \ref{sec:biochemicalexamples}, we give biochemically motivated examples without interpretations as models of infectious disease spread.

\subsection{Infectious Disease Models}
\label{sec:infectiousexamples}

Here, we consider several examples from the mathematical epidemiology literature. A recurring theme is that a computationally efficient choice of $\mathcal{F}$ uses heuristic \textbf{(H1)}, since infection events where contact between a susceptible and infectious individual transmits infection can be represented as autocatalytic reactions $S + I \to 2I$.

\begin{example}
Consider the SIR (Susceptible–Infectious–Recovered) disease model with births, deaths, and a saturating incidence rate \cite{VANDENDRIESSCHE2002,rahman2012global}. For simplicity, we omit the time delays present in \cite{rahman2012global}. This can be represented in the interaction form used in mathematical biochemistry as:
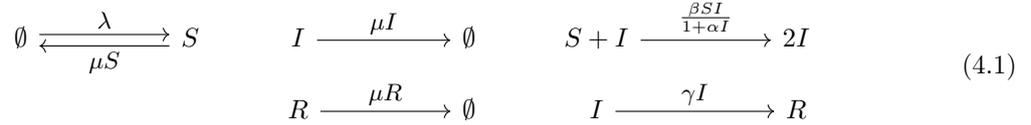
\begin{equation}
\label{SIR}
\begin{tikzcd}
\emptyset \arrow[rr,yshift=0.5ex,"\lambda"] & & S \arrow[ll,yshift=-0.5ex,"\mu S"] & I \arrow[rr,"\mu I"] & & \emptyset & S + I \arrow[rr,"\frac{\beta SI}{1 + \alpha I}"] & & 2I \\[-0.1in]
 & & & R \arrow[rr,"\mu R"] & & \emptyset & I \arrow[rr,"\gamma I"] & & R
\end{tikzcd}
\end{equation}
where all reaction rates are mass-action except the saturating incidence rate, which satisfies \textbf{(A1)}–\textbf{(A3)}. The network \eqref{SIR} corresponds to the system of differential equations
\[
\begin{cases}
S' = \displaystyle \lambda - \frac{\beta SI}{1 + \alpha I} - \mu S, \\
I' = \displaystyle \frac{\beta SI}{1 + \alpha I} - (\gamma + \mu) I, \\
R' = \gamma I - \mu R.
\end{cases}
\]
The system has the siphon $\mathcal{X} = \{ I, R \}$, which is critical. There are no conservation laws, and the only $\mathcal{X}$-free boundary steady state is $\displaystyle \mathbf{x}^* = \left(\frac{\lambda}{\mu}, 0, 0\right)$.

Following heuristic \textbf{(H1)}, we choose 
\[
\mathcal{F} = \left\langle \frac{\beta SI}{1 + \alpha I}, 0 \right\rangle, \quad
\mathcal{V} = \langle (\gamma + \mu) I, -\gamma I + \mu R \rangle.
\]
Taking the Jacobians of $\mathcal{F}$ and $\mathcal{V}$ with respect to $\mathcal{X} = \{I, R\}$ and evaluating at $\mathbf{x}^*$ gives
\[
F = \begin{bmatrix}
\displaystyle \frac{\lambda \beta}{\mu} & 0 \\
0 & 0
\end{bmatrix}, \quad
V = \begin{bmatrix}
\gamma + \mu & 0 \\
-\gamma & \mu
\end{bmatrix}, \quad
V^{-1} = \begin{bmatrix}
\displaystyle \frac{1}{\gamma + \mu} & 0 \\
\displaystyle \frac{\gamma}{(\gamma + \mu)\mu} & \displaystyle \frac{1}{\mu}
\end{bmatrix}, \quad
F V^{-1} = \begin{bmatrix}
\displaystyle \frac{\lambda \beta}{(\gamma + \mu)\mu} & 0 \\
0 & 0
\end{bmatrix}.
\]
Thus, $V$ satisfies the requirements of Theorem \ref{thm:main}. The boundary reproduction number of the $\mathcal{X}$-free boundary steady state $\mathbf{x}^*$ is $\mathscr{R}_{\mathbf{x}^*} = \rho(FV^{-1}) = \displaystyle \frac{\lambda \beta}{(\gamma + \mu)\mu}$. Hence, $\mathbf{x}^*$ is locally asymptotically stable if $\mathscr{R}_{\mathbf{x}^*} < 1$ and unstable if $\mathscr{R}_{\mathbf{x}^*} > 1$. In the context of infectious disease models, $\mathbf{x}^*$ corresponds to the disease-free state; the disease dies out if $\mathscr{R}_{\mathbf{x}^*} < 1$ (stable disease-free state) and spreads if $\mathscr{R}_{\mathbf{x}^*} > 1$ (unstable disease-free state). Note that $\tilde{\mathcal{X}} = \{ I \}$ is also a siphon with $\mathbf{x}^* \in E_\mathcal{X}$. Analyzing this siphon yields the same boundary reproduction number and stability conditions.
\end{example}

\begin{example}[Connection with Multistrain Models]
Consider the infectious disease model with two strains, generically labeled $I_1$ and $I_2$, with births, deaths, and saturating incidence rates \cite{rahman2012global}. We omit time delays. Such a model can be represented as a chemical reaction network with the following interaction diagram:
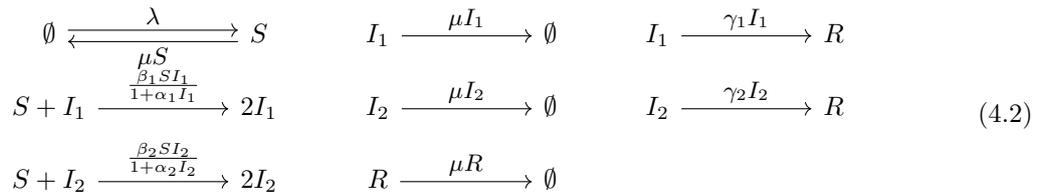
\begin{equation}
\label{multistrain}
\begin{tikzcd}
\emptyset \arrow[rr,yshift=0.5ex,"\lambda"] & & S \arrow[ll,yshift=-0.5ex,"\mu S"] & I_1 \arrow[rr,"\mu I_1"] & & \emptyset & I_1 \arrow[rr,"\gamma_1 I_1"] & & R \\[-0.1in]
S + I_1 \arrow[rr,"\frac{\beta_1 SI_1}{1 +\alpha_1 I_1}"] & & 2I_1 & I_2 \arrow[rr,"\mu I_2"] & & \emptyset & I_2 \arrow[rr,"\gamma_2 I_2"] & & R\\[-0.1in]
S + I_2 \arrow[rr,"\frac{\beta_2 SI_2}{1 + \alpha_2 I_2}"] & & 2I_2 & R \arrow[rr,"\mu R"] & & \emptyset & & & \\[-0.1in]
\end{tikzcd}
\end{equation}
All reaction rates are mass-action except the incidence rates. The rates $\displaystyle \frac{\beta_1 SI_1}{1 +\alpha_1 I_1}$ and $\displaystyle \frac{\beta_2 SI_2}{1 + \alpha_2 I_2}$ satisfy assumptions \textbf{(A1)}--\textbf{(A3)}. The network \eqref{multistrain} corresponds to the system of differential equations
\[
\left\{
\begin{aligned}
S' & = \lambda - \frac{\beta_1 SI_1}{1 +\alpha_1 I_1} - \frac{\beta_2 SI_2}{1 + \alpha_2 I_2} - \mu S, \\
I_1' & = \frac{\beta_1 SI_1}{1 +\alpha_1 I_1} - (\gamma_1 + \mu) I_1, \\
I_2' & = \frac{\beta_2 SI_2}{1 + \alpha_2 I_2} - (\gamma_2 + \mu) I_2, \\
R' & = \gamma_1 I_1 + \gamma_2 I_2 - \mu R.
\end{aligned}
\right.
\]
Three siphons for this model are $\mathcal{X}_1 = \{ I_1, I_2, R \}$, $\mathcal{X}_2 = \{ I_1 \}$, and $\mathcal{X}_3 = \{ I_2 \}$. All three siphons are critical, so we compute the boundary reproduction numbers with respect to specific steady states rather than the siphons. There are no conservation laws. The $\mathcal{X}_1$-free boundary steady state is
\[
\mathbf{x}_1^* = \left(\frac{\lambda}{\mu}, 0, 0, 0\right),
\]
the $\mathcal{X}_2$-free boundary steady state is
\[
\mathbf{x}_2^* = \left( \frac{\lambda \alpha_2 + \gamma_2 + \mu}{\alpha_2 \mu + \beta_2}, 0 , \frac{\lambda \beta_2 - (\gamma_2 + \mu)\mu}{\alpha_2 \gamma_2 \mu + \alpha_2 \mu^2 + \beta_2 \gamma_2 + \beta_2 \mu}, \frac{\gamma_2 (\lambda \beta_2 - (\gamma_2 + \mu)\mu)}{(\alpha_2 \gamma_2 \mu + \alpha_2 \mu^2 + \beta_2 \gamma_2 + \beta_2 \mu) \mu} \right),
\]
which is nonnegative if $\lambda \beta_2 > (\gamma_2 + \mu)\mu$, and the $\mathcal{X}_3$-free boundary steady state
\[
\mathbf{x}_3^* = \left( \frac{\lambda \alpha_1 + \gamma_1 + \mu}{\alpha_1 \mu + \beta_1}, \frac{\lambda \beta_1 - (\gamma_1 + \mu)\mu}{\alpha_1 \gamma_1 \mu + \alpha_1 \mu^2 + \beta_1 \gamma_1 + \beta_1 \mu}, 0, \frac{\gamma_1 (\lambda \beta_1 - (\gamma_1 + \mu)\mu)}{(\alpha_1 \gamma_1 \mu + \alpha_1 \mu^2 + \beta_1 \gamma_1 + \beta_1 \mu)\mu} \right),
\]
which is nonnegative if $\lambda \beta_1 > (\gamma_1 + \mu)\mu$.

We now compute the boundary reproduction numbers of the three steady states. In all cases, we use heuristic \textbf{(H1)} and include terms corresponding to all relevant infection events (i.e., autocatalytic reactions) in $\mathcal{F}$. For the $\mathcal{X}_1$-free boundary steady state $\mathbf{x}_1^*$, we have
\[
\begin{aligned}
\mathcal{F}_1 & = \Big\langle \frac{\beta_1 SI_1}{1 +\alpha_1 I_1}, \frac{\beta_2 SI_2}{1 + \alpha_2 I_2} , 0 \Big\rangle, \\
\mathcal{V}_1 & = \langle  (\gamma_1 + \mu) I_1, (\gamma_2 + \mu) I_2, -\gamma_1 I_1 - \gamma_2 I_2 + \mu R  \rangle,
\end{aligned}
\]
so that
\[
F_1 = \begin{bmatrix}
\displaystyle\frac{\lambda \beta_1}{\mu} & 0 & 0 \\
0 & \displaystyle\frac{\lambda \beta_2}{\mu} & 0 \\
0 & 0 & 0
\end{bmatrix}
\quad \text{and} \quad
V_1 = \begin{bmatrix}
\gamma_1 + \mu & 0 & 0 \\
0 & \gamma_2 + \mu & 0 \\
-\gamma_1 & - \gamma_2 & \mu
\end{bmatrix}.
\]
It follows that
\[
V_1^{-1} = \begin{bmatrix}
\displaystyle\frac{1}{\gamma_1 + \mu} & 0 & 0 \\
0 & \displaystyle\frac{1}{\gamma_2 + \mu} & 0 \\
\displaystyle\frac{\gamma_1}{(\gamma_1 + \mu) \mu} & \displaystyle\frac{\gamma_2}{(\gamma_2 + \mu) \mu} & \displaystyle\frac{1}{\mu}
\end{bmatrix}
\quad \text{and} \quad
F_1 V_1^{-1} = \begin{bmatrix}
\displaystyle\frac{\lambda \beta_1}{\mu(\gamma_1 + \mu)} & 0 & 0 \\
0 & \displaystyle\frac{\lambda \beta_2}{\mu(\gamma_2 + \mu)} & 0 \\
0 & 0 & 0
\end{bmatrix}.
\]
Consequently,
\[
\mathscr{R}_{\mathbf{x}_1^*} = \rho(F_1 V_1^{-1}) = \max \left\{ \frac{\lambda \beta_1}{\mu(\gamma_1 + \mu)} , \frac{\lambda \beta_2}{\mu(\gamma_2 + \mu)} \right\} = \max \{ \mathscr{R}_{\mathbf{x}_1^*}^{(1)}, \mathscr{R}_{\mathbf{x}_1^*}^{(2)} \}.
\]
Hence, $\mathbf{x}_1^*$ is locally asymptotically stable if both $\mathscr{R}_{\mathbf{x}_1^*}^{(1)} < 1$ and $\mathscr{R}_{\mathbf{x}_1^*}^{(2)} < 1$, and unstable if either $\mathscr{R}_{\mathbf{x}_1^*}^{(1)} > 1$ or $\mathscr{R}_{\mathbf{x}_1^*}^{(2)} > 1$. Note that $\tilde{\mathcal{X}}_1 = \{ I_1, I_2 \}$ is also a siphon with $\mathbf{x}_1^* \in E_{\mathcal{X}}$, and this choice yields the same boundary reproduction number and stability regions as using $\mathcal{X}_1 = \{ I_1, I_2, R \}$.

For the $\mathcal{X}_2$-free boundary steady state $\mathbf{x}_2^*$, we have
\[
\mathcal{F}_2 = \frac{\beta_1 SI_1}{1 + \alpha_1 I_1}, \quad \mathcal{V}_2 = (\gamma_1 + \mu) I_1, \quad F_2 = \frac{\beta_1 (\lambda \alpha_2 + \gamma_2 + \mu)}{\alpha_2 \mu + \beta_2}, \quad V_2 = \gamma_1 + \mu,
\]
so that
\[
\mathscr{R}_{\mathbf{x}_2^*} = \frac{\beta_1 (\lambda \alpha_2 + \gamma_2 + \mu)}{(\alpha_2 \mu + \beta_2)(\gamma_1 + \mu)}.
\]
Hence, $\mathbf{x}_2^*$ is locally asymptotically stable if $\mathscr{R}_{\mathbf{x}_2^*} < 1$ and unstable if $\mathscr{R}_{\mathbf{x}_2^*} > 1$.

For the $\mathcal{X}_3$-free boundary steady state $\mathbf{x}_3^*$, we have
\[
\mathcal{F}_3 = \frac{\beta_2 SI_2}{1 + \alpha_2 I_2}, \quad \mathcal{V}_3 = (\gamma_2 + \mu) I_2, \quad F_3 = \frac{\beta_2 (\lambda \alpha_1 + \gamma_1 + \mu)}{\alpha_1 \mu + \beta_1}, \quad V_3 = \gamma_2 + \mu,
\]
so that
\[
\mathscr{R}_{\mathbf{x}_3^*} = \frac{\beta_2 (\lambda \alpha_1 + \gamma_1 + \mu)}{(\alpha_1 \mu + \beta_1)(\gamma_2 + \mu)}.
\]
Hence, $\mathbf{x}_3^*$ is locally asymptotically stable if $\mathscr{R}_{\mathbf{x}_3^*} < 1$ and unstable if $\mathscr{R}_{\mathbf{x}_3^*} > 1$.

In the context of multistrain infectious disease models, $\mathbf{x}_1^*$ corresponds to the disease-free state, $\mathbf{x}_2^*$ to the $I_1$-free/$I_2$-endemic state, and $\mathbf{x}_3^*$ to the $I_1$-endemic/$I_2$-free state. The value $\mathscr{R}_{\mathbf{x}_1^*}$ is the basic reproduction number of the disease since it determines whether the disease (either strain) will spread. The values $\mathscr{R}_{\mathbf{x}_1^*}^{(1)}$ and $\mathscr{R}_{\mathbf{x}_1^*}^{(2)}$ are the basic reproduction numbers for the individual strains; that is, they determine whether the introduction of each strain can destabilize the disease-free state and spread independently of the other strain. Note also that $\mathscr{R}_{\mathbf{x}_1^*}^{(1)} > 1$ is required for $\mathbf{x}_2^* \geq 0$ and $\mathscr{R}_{\mathbf{x}_1^*}^{(2)} > 1$ is required for $\mathbf{x}_3^* \geq 0$. Thus, the single-strain endemic steady states are biologically relevant only if the corresponding strain-specific reproduction number exceeds one. The values $\mathscr{R}_{\mathbf{x}_2^*}$ and $\mathscr{R}_{\mathbf{x}_3^*}$ are invasion numbers: they indicate when the introduction of one strain can destabilize the steady state endemic in the other strain and invade the population \cite{martcheva2007,martcheva2008,johnston2025effect}. Coexistence of both strains is possible if $\mathscr{R}_{\mathbf{x}_1^*} > 1$, $\mathscr{R}_{\mathbf{x}_2^*} > 1$, and $\mathscr{R}_{\mathbf{x}_3^*} > 1$, i.e., either strain can invade the other.
\end{example}

\begin{example}
\label{ex:vectorborne}
Consider the vector-borne disease spread model introduced in \cite{feng1997competitive} and studied further in \cite{VANDENDRIESSCHE2002}.
%\begin{equation}
%\label{vector-de}
\[\left\{ \; \;
\begin{aligned}
S' & = \alpha_s-\delta_s S+\gamma_s I - \beta_s S V \\
I' & = \beta_s S V - (\delta_s + \gamma_s)I\\
M' & = \alpha_m-\delta_m M - \beta_m M I \\
V' & = \beta_m M I - \delta_m V.
\end{aligned}
\right.\]
%\end{equation}
This can be represented as the following chemical interaction network:
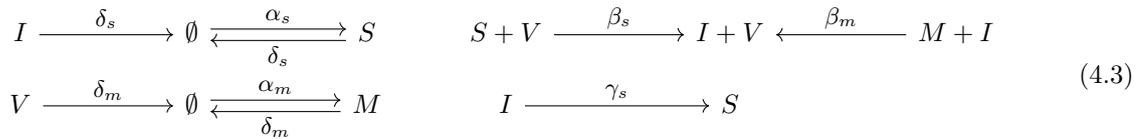
\begin{equation}
\label{vector}
\begin{tikzcd}
I \arrow[rr,"\delta_s"] & & \emptyset \arrow[rr,yshift=0.5ex,"\alpha_s"] & & S \arrow[ll,yshift=-0.5ex,"\delta_s"] & S + V \arrow[rr,"\beta_s"] & & I + V & & M + I \arrow[ll,"\beta_m"'] \\[-0.1in]
V \arrow[rr,"\delta_m"] & & \emptyset \arrow[rr,yshift=0.5ex,"\alpha_m"] & & M \arrow[ll,yshift=-0.5ex,"\delta_m"] & I \arrow[rr,"\gamma_s"] & & S & & \\[-0.1in]
\end{tikzcd}
\end{equation}

The network \eqref{vector} has only the siphon $\mathcal{X} = \{ V, I \}$, which is critical. There are no conservation laws. The corresponding $\mathcal{X}$-free boundary steady state is $\displaystyle \mathbf{x}^* = \left(\frac{\alpha_s}{\delta_s}, 0, \frac{\alpha_m}{\delta_m}, 0 \right)$. For heuristic \textbf{(H1)}, we notice that $S + V \to I + V$ has $V$ as both a reactant and product, and $M+I \to I+V$ has $I$ as both a reactant and product, but neither species is increased. Consequently, \textbf{(H1)} does not apply. However, we can use heuristic \textbf{(H2)} since both $S+V \to I +V$ and and $M+I \to I + V$ have $I$ and $V$ as products. %This heuristic allows us to choose to include either $\beta_s S V$ or $\beta_m M I$ in $\mathcal{F}$ and test them until a valid combination (i.e one for which $V^{-1} \geq 0$) is found. 
Choosing the term $\beta_m M I$ in the $V'$ equation to include in $\mathcal{F}$ gives:% either term individually yields $V^{-1} \geq 0$ and the following boundary reproduction number:
\[
F =\begin{bmatrix}
0 & 0 \\
\displaystyle \frac{\alpha_m \beta_m}{\delta_m} & 0
\end{bmatrix}, \quad
V = \begin{bmatrix}
\delta_s + \gamma_s & \displaystyle - \frac{\alpha_s \beta_s}{\delta_s} \\
0  & \delta_m
\end{bmatrix}, \quad
FV^{-1} = \begin{bmatrix}
0 & 0 \\
\displaystyle \frac{\alpha_m \beta_m}{\delta_m (\delta_s + \gamma_s)} & \displaystyle \frac{\alpha_m \beta_m \beta_s \alpha_s}{\delta_m^2 \delta_s (\delta_s + \gamma_s)}
\end{bmatrix}.
\]
$V$ is in the block triangular form \eqref{blocktriangular} with $A_1 = \delta_s + \gamma_s >0$ and $A_2 = \delta_m > 0$. It follows that 
\begin{equation}
\label{bpn2}
\mathscr{R}_{\mathbf{x}^*} = \rho(FV^{-1}) = \frac{\beta_s \alpha_s \beta_m \alpha_m}{\delta_s \delta_m^2 (\delta_s + \gamma_s)}
\end{equation}
so that $\mathbf{x}^*$ is locally asymptotically stable if $\mathscr{R}_{\mathbf{x}^*}<1$ and unstable if $\mathscr{R}_{\mathbf{x}^*}>1$.

It is worth noting that there may be multiple choices of $\mathcal{F}$ which satisfy Theorem \ref{thm:main}, and they may produce different values for $\mathscr{R}_{\mathbf{x}^*}$. For example, consider using both the term $\beta_s S V$ in the $I'$ equation and the term $\beta_m M I$ in the $V'$ equation to include in $\mathcal{F}$. This choice gives
\[
F =\begin{bmatrix}
0 & \displaystyle  \frac{\alpha_s\beta_s}{\delta_s} \\
\displaystyle \frac{\alpha_m\beta_m}{\delta_m}  & 0
\end{bmatrix}, \quad
V = \begin{bmatrix}
\delta_s + \gamma_s & 0 \\
0 & \delta_m
\end{bmatrix}, \quad
FV^{-1} = \begin{bmatrix}
0 & \displaystyle \frac{\beta_s \alpha_s}{\delta_s \delta_m} \\
\displaystyle \frac{\alpha_m \beta_m}{\delta_m (\delta_s + \gamma_s)} & 0
\end{bmatrix}
\]
so that
\begin{equation}
\label{bpn3}
\mathscr{R}_{\mathbf{x}^*} = \rho(FV^{-1}) = \sqrt{\frac{\beta_s \alpha_s \beta_m \alpha_m}{\delta_s \delta_m^2 (\delta_s + \gamma_s)}}.
\end{equation}
Note that, although the choice of $\mathcal{F}$ produces different boundary reproduction numbers \eqref{bpn2} and \eqref{bpn3}, they both yield the same bifurcation value $\mathscr{R}_{\mathbf{x}^*} = 1$ and are therefore consistent with Theorem \ref{thm:main}.
%Note that, following either heuristic \textbf{(H1)} or \textbf{(H2)}, we have the same boundary threshold for stability of the disease-free state given by $\mathscr{R}_{\mathbf{x}^*} = 1$. Note that even though the choice of $\mathcal{F}$ produces different values of $\mathscr{R}_{\mathbf{x}^*}$, the same critical threshold $\mathscr{R}_{\mathbf{x}^*} = 1$ is obtained.
%{\color{red}(Vector-borne illness example here)}
\end{example}

\subsection{Biochemical Reaction Networks}
\label{sec:biochemicalexamples}

In this section, we present several examples with boundary steady states drawn from the mathematical biochemistry literature. A key theme is an increased reliance on the disassociation heuristic \textbf{(H2)} rather than the autocatalysis heuristic \textbf{(H1)} in the selection of $\mathcal{F}$. Conservation laws factor significantly in the analysis.

\begin{example}
\label{ex:IDHKP}
Consider the following model of the IDHKP-IDH system presented in \cite{shinar2010structural} and previously analyzed in \cite{Avram2025StabilityIR}:
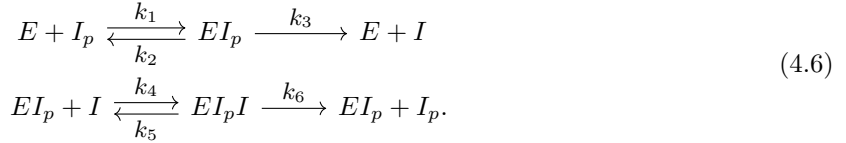
\begin{equation}
\label{IDHKP1}
\begin{tikzcd}
E + I_p \arrow[r,yshift=0.5ex,"k_1"] & EI_p \arrow[l,yshift=-0.5ex,"k_2"] \arrow[r,"k_3"] & E + I \\[-0.1in]
EI_p + I \arrow[r,yshift=0.5ex,"k_4"] & EI_pI \arrow[l,yshift=-0.5ex,"k_5"] \arrow[r,"k_6"] & EI_p + I_p.
\end{tikzcd}
\end{equation}
The species set is $\mathcal{S} = \{ E, I_p, EI_p, I, EI_pI \}$, and $\mathcal{X} = \{I_p, EI_p, EI_pI \}$ is a critical siphon. This motivates the relabeling $X_1 = I_p$, $X_2 = EI_p$, $X_3 = EI_pI$, $Y_1 = E$, and $Y_2 = I$, so that we have the siphon $\mathcal{X} = \{ X_1, X_2, X_3 \}$ and antisiphon $\mathcal{Y} = \{ Y_1, Y_2 \}$. We rewrite \eqref{IDHKP1} as
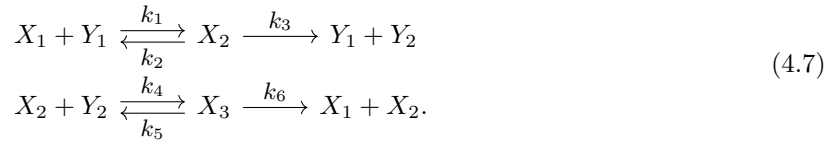
\begin{equation}
\label{IDHKP2}
\begin{tikzcd}
X_1 + Y_1 \arrow[r,yshift=0.5ex,"k_1"] & X_2 \arrow[l,yshift=-0.5ex,"k_2"] \arrow[r,"k_3"] & Y_1 + Y_2 \\[-0.1in]
X_2 + Y_2 \arrow[r,yshift=0.5ex,"k_4"] & X_3 \arrow[l,yshift=-0.5ex,"k_5"] \arrow[r,"k_6"] & X_1 + X_2.
\end{tikzcd}
\end{equation}
The mass-action system \eqref{de2} corresponding to \eqref{IDHKP2} is
\begin{equation}
\label{IDHKP-de}
\left\{ \; \;
\begin{aligned}
x_1' & = -k_1x_1y_1 + k_2x_2 + k_6x_3\\
x_2' & = k_1x_1y_1 - (k_2 + k_3)x_2 - k_4x_2y_2 + (k_5 + k_6)x_3\\
x_3' & = k_4x_2y_2 - (k_5 + k_6)x_3
\end{aligned}
\right. \; \; \left\{ \; \;
\begin{aligned}
y_1' & = -k_1x_1y_1 + (k_2 + k_3)x_2\\
y_2' & = -k_4x_2y_2 + k_3x_2 + k_5x_3
\end{aligned}
\right.
\end{equation}
The network \eqref{IDHKP2} admits the conservation laws $X_2 + X_3 + Y_1 = E_{tot}$ and $X_1 + X_2 + 2X_3 + Y_2 = I_{tot}$, with positive conservation constants $E_{tot}, I_{tot} > 0$, corresponding to the conservation equation \eqref{conservation2} with
\[
W_{\mathcal{X}} = \begin{bmatrix} 0 & 1 & 1 \\ 1 & 1 & 2 \end{bmatrix}, \quad 
W_{\mathcal{Y}} = \begin{bmatrix} 1 & 0 \\ 0 & 1 \end{bmatrix}, \quad 
\text{and} \quad 
\mathbf{\Lambda} = \begin{bmatrix} E_{tot} \\ I_{tot} \end{bmatrix}.
\]
We have $|\mathcal{Y}| = n-s = 2$ and $W_{\mathcal{Y}}$ is invertible. Consequently, every $\mathcal{X}$-free boundary steady state has the form $\mathbf{x}^* = W_{\mathcal{Y}}^{-1}\mathbf{\Lambda} = (0,0,0,E_{tot},I_{tot})$.

We now select $\mathcal{F}$. Since there are no autocatalytic reactions, heuristic \textbf{(H1)} does not apply. However, the disassociation reaction $X_3 \to X_1 + X_2$ suggests, via heuristic \textbf{(H2)}, including the term $k_6x_3$ in the $x_1'$ or $x_2'$ equation of \eqref{IDHKP-de} in $\mathcal{F}$. Choosing the $k_6x_3$ term in the $x_1'$ equation gives:
\[
\begin{aligned}
\mathcal{F} & = \langle k_6x_3, 0, 0 \rangle,\\
\mathcal{V} & = \langle k_1x_1y_1 - k_2x_2, -k_1x_1y_1 + (k_2 + k_3)x_2 + k_4x_2y_2 - (k_5 + k_6)x_3, (k_5 + k_6)x_3 - k_4x_2y_2 \rangle.
\end{aligned}
\]
We take the Jacobians of $\mathcal{F}$ and $\mathcal{V}$ with respect to the siphon elements $\mathcal{X} = \{ X_1, X_2, X_3 \}$ and evaluate them at the boundary steady state $\mathbf{x}^*$ to obtain
\[
F = \begin{bmatrix}
0 & 0 & k_6 \\
0 & 0 & 0 \\
0 & 0 & 0 
\end{bmatrix}
\quad \text{and} \quad
V = \begin{bmatrix}
k_1 E_{tot} & -k_2 & 0 \\
-k_1 E_{tot} & k_4 I_{tot} + k_2 + k_3 & -k_5 - k_6 \\
0 & -k_4 I_{tot} & k_5 + k_6
\end{bmatrix}.
\]
Note that $V$ is irreducible (i.e., $V = A_1$ in \eqref{blocktriangular}) and $\mathbbm{1}^T V = (0, k_3, 0) \geq 0$. Thus, the conditions of Theorem \ref{thm:2} are satisfied, implying $V^{-1} \geq 0$. We are therefore justified in computing the boundary reproduction number $\mathscr{R}_{\mathbf{x}^*} = \rho(FV^{-1})$ by Theorem \ref{thm:main}.

We compute:
\[
V^{-1} = \left[ \begin{array}{ccc} 
\displaystyle \frac{k_2 + k_3}{k_1 k_3E_{tot}} & \displaystyle \frac{k_2}{k_1 k_3E_{tot}} & \displaystyle \frac{k_2}{k_1 k_3E_{tot}} \\[1em]
\displaystyle \frac{1}{k_3} & \displaystyle \frac{1}{k_3} & \displaystyle \frac{1}{k_3} \\[1em]
\displaystyle \frac{k_4I_{tot}}{k_3(k_5 + k_6)} & \displaystyle \frac{k_4I_{tot}}{k_3(k_5 + k_6)} & \displaystyle \frac{k_4I_{tot} + k_3}{k_3(k_5 + k_6)} 
\end{array} \right]
\]
so that
\[
FV^{-1} = 
\left[ \begin{array}{ccc}
\displaystyle \frac{k_4 k_6I_{tot}}{k_3(k_5 + k_6)} & \displaystyle \frac{k_4 k_6I_{tot}}{k_3(k_5 + k_6)} & \displaystyle \frac{k_6(k_4I_{tot} + k_3)}{k_3(k_5 + k_6)} \\
0 & 0 & 0 \\
0 & 0 & 0
\end{array} \right].
\]
The boundary reproduction number of $\mathbf{x}^*$ for this choice of splitting $\mathcal{F}$ and $\mathcal{V}$ is
\[
\mathscr{R}_{\mathbf{x}^*} = \rho(FV^{-1}) = \displaystyle{\frac{k_4k_6 I_{tot}}{k_3 (k_5 + k_6)}}.
\]
It follows that $\mathbf{x}^*$ is locally asymptotically stable for every stoichiometric compatibility class satisfying $\displaystyle I_{tot} < \frac{k_3(k_5 + k_6)}{k_4k_6}$ and unstable for every compatibility class satisfying $\displaystyle I_{tot} > \frac{k_3(k_5 + k_6)}{k_4k_6}$.
\end{example}

\begin{example}
\label{ex:EnvZOmpR19}
Consider the full EnvZ-OmpR network introduced in the Supplemental Material of \cite{shinar2010structural} and previously analyzed in \cite{Avram2025StabilityIR}:
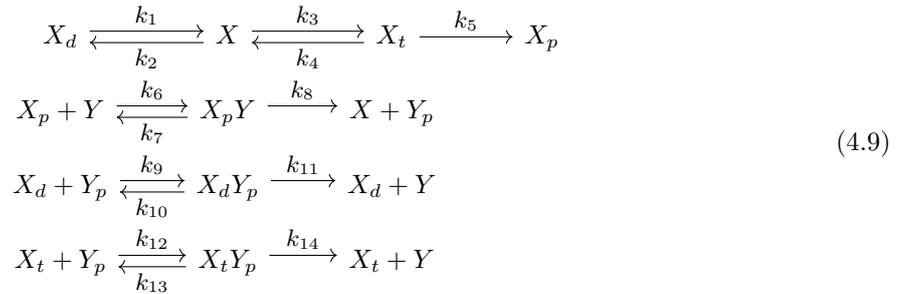
\begin{equation}
\label{EnvZOmpR}
      \begin{tikzcd}
X_d \arrow[r,yshift=+0.5ex,"k_1"] & X \arrow[l,yshift=-0.5ex,"k_2"] \arrow[r,yshift=+0.5ex,"k_3"]& X_t \arrow[l,yshift=-0.5ex,"k_4"] \arrow[r,"k_5"]& X_p\\[-0.1in]
X_p + Y \arrow[r,yshift=+0.5ex,"k_6"] & X_pY \arrow[l,yshift=-0.5ex,"k_7"] \arrow[r,yshift=+0.5ex,"k_8"]& X + Y_p & \\[-0.1in]
X_d + Y_p \arrow[r,yshift=+0.5ex,"k_9"] & X_dY_p \arrow[l,yshift=-0.5ex,"k_{10}"] \arrow[r,yshift=+0.5ex,"k_{11}"]& X_d + Y &  \\[-0.1in]
X_t + Y_p \arrow[r,yshift=+0.5ex,"k_{12}"] & X_tY_p \arrow[l,yshift=-0.5ex,"k_{13}"] \arrow[r,yshift=+0.5ex,"k_{14}"]& X_t + Y &
      \end{tikzcd}
\end{equation}
We have the siphon $\mathcal{X} = \{ X_d, X, X_t, Y, X_pY, X_dY_p, X_tY_p \}$, which is critical. This suggests the relabeling $X_1 = X_d, X_2 = X, X_3 = X_t, X_4 = Y, X_5 = X_pY, X_6 = X_dY_p, X_7 = X_tY_p, Y_1 = X_p$, and $Y_2 = Y_p$ in \eqref{EnvZOmpR}, giving:
\begin{equation}
\label{EnvZOmpR2}
      \begin{tikzcd}
X_1 \arrow[r,yshift=+0.5ex,"k_1"] & X_2 \arrow[l,yshift=-0.5ex,"k_2"] \arrow[r,yshift=+0.5ex,"k_3"]& X_3 \arrow[l,yshift=-0.5ex,"k_4"] \arrow[r,"k_5"]& Y_1\\[-0.1in]
X_4 + Y_1 \arrow[r,yshift=+0.5ex,"k_6"] & X_5 \arrow[l,yshift=-0.5ex,"k_7"] \arrow[r,yshift=+0.5ex,"k_8"]& X_2 + Y_2 & \\[-0.1in]
X_1 + Y_2 \arrow[r,yshift=+0.5ex,"k_9"] & X_6 \arrow[l,yshift=-0.5ex,"k_{10}"] \arrow[r,yshift=+0.5ex,"k_{11}"]& X_1 + X_4 &  \\[-0.1in]
X_3 + Y_2 \arrow[r,yshift=+0.5ex,"k_{12}"] & X_7 \arrow[l,yshift=-0.5ex,"k_{13}"] \arrow[r,yshift=+0.5ex,"k_{14}"]& X_3 + X_4 &
      \end{tikzcd}
\end{equation}
The associated mass-action system is:
\small
\begin{equation}
\label{EnvZOmpR-de}
\left\{ \; \;
\begin{aligned}
x_1' & =  -k_9 x_1 y_2 - k_1 x_1 + k_2 x_2 + (k_{10} + k_{11}) x_6 \\
x_2' & =  k_1 x_1 - (k_2 + k_3) x_2 + k_4 x_3 + k_8 x_5\\
x_3' & =  k_3 x_2 - (k_4 + k_5) x_3 - k_{12} x_3 y_2 + (k_{13} + k_{14}) x_7 \\
x_4' & =  -k_6 x_4 y_1 + k_5 x_3 + k_7 x_5\\
x_5' & =  k_6 x_4 y_1 - (k_7 + k_8) x_5\\
x_6' & =  k_9 x_1 y_2 - (k_{10} + k_{11}) x_6\\
x_7' & =  k_{12} x_3 y_2 - (k_{13} + k_{14}) x_7
\end{aligned}
\right. \; \; \left\{ \; \; 
\begin{aligned}
y_1' & =  -k_6 x_4 y_1 + k_{11} x_6 + k_{14} x_7 + k_7 x_5\\
y_2' & =  -k_{12} x_3 y_2 - k_9 x_1 y_2 + k_{10} x_6 + k_{13} x_7 + k_8 x_5
\end{aligned} \right.
\end{equation}
\normalsize
The network \eqref{EnvZOmpR2} has the conservation laws $X_1+X_2+X_3+X_5+X_6+X_7+Y_1 = X_{tot}$ and $X_4+X_5+X_6+X_7+Y_2 = Y_{tot}$, with conservation constants $X_{tot} > 0$ and $Y_{tot} > 0$, corresponding to the conservation equation \eqref{conservation2}, where
\[W_{\mathcal{X}} = \begin{bmatrix} 1 & 1 & 1 & 0 & 1 & 1 & 1 \\ 0 & 0 & 0 & 1 & 1 & 1 & 1 \end{bmatrix}, \quad W_{\mathcal{Y}} = \begin{bmatrix} 1 & 0 \\ 0 & 1 \end{bmatrix}, \quad \text{and} \quad \mathbf{\Lambda} = \begin{bmatrix} X_{tot} \\ Y_{tot} \end{bmatrix}.\]
We have $|\mathcal{Y}| = n-s = 2$ and $W_{\mathcal{Y}}$ is invertible, so every $\mathcal{X}$-free boundary steady state has the form $\mathbf{x}^* = W_{\mathcal{Y}}^{-1} \mathbf{\Lambda} = (0,0,0,0,0,0,0,X_{tot},Y_{tot})$.

We now turn to selecting the positive terms to include in $\mathcal{F}$. There are no autocatalytic reactions, so heuristic \textbf{(H1)} does not apply. However, there are two disassociation reactions, $X_6 \stackrel{k_{11}}{\longrightarrow} X_1 + X_4$ and $X_7 \stackrel{k_{14}}{\longrightarrow} X_3 + X_4$, so heuristic \textbf{(H2)} applies. After some trial and error, we select the $k_{11}x_6$ and $k_{14}x_7$ terms in the $x_4'$ equation of \eqref{EnvZOmpR-de}. This gives:
\[
\mathcal{F} = \langle 0,0,0, k_{11} x_6 + k_{14} x_7, 0, 0, 0 \rangle,
\]
and $\mathcal{V} = \langle \mathcal{V}_1, \mathcal{V}_2, \mathcal{V}_3, \mathcal{V}_4, \mathcal{V}_5, \mathcal{V}_6, \mathcal{V}_7 \rangle$ with entries
\[
\begin{aligned}
\mathcal{V}_1 & = k_9 x_1 y_2 + k_1 x_1 - k_2 x_2 - (k_{10} + k_{11}) x_6,\\
\mathcal{V}_2 & = -k_1 x_1 + (k_2 + k_3) x_2 - k_4 x_3 - k_8 x_5,\\
\mathcal{V}_3 & = -k_3 x_2 + (k_4 + k_5) x_3 + k_{12} x_3 y_2 - (k_{13} + k_{14}) x_7,\\
\mathcal{V}_4 & = k_6 x_4 y_1 - k_7 x_5,\\
\mathcal{V}_5 & = -k_6 x_4 y_1 + (k_7 + k_8) x_5,\\
\mathcal{V}_6 & = -k_9 x_1 y_2 + (k_{10} + k_{11}) x_6,\\
\mathcal{V}_7 & = -k_{12} x_3 y_2 + (k_{13} + k_{14}) x_7.\\
\end{aligned}
\]
We take the Jacobians of $\mathcal{F}$ and $\mathcal{V}$ with respect to the siphon elements $\mathcal{X}$ and evaluate at $\mathbf{x}^*$ to get
\[F = \left[ \begin{array}{ccccccc}
\displaystyle 0 & \displaystyle 0 & \displaystyle 0 & \displaystyle 0 & \displaystyle 0 & \displaystyle 0 & \displaystyle 0 \\
\displaystyle 0 & \displaystyle 0 & \displaystyle 0 & \displaystyle 0 & \displaystyle 0 & \displaystyle 0 & \displaystyle 0 \\
\displaystyle 0 & \displaystyle 0 & \displaystyle 0 & \displaystyle 0 & \displaystyle 0 & \displaystyle 0 & \displaystyle 0 \\
\displaystyle 0 & \displaystyle 0 & \displaystyle 0 & \displaystyle 0 & \displaystyle 0 & \displaystyle k_{11} & \displaystyle k_{14} \\
\displaystyle 0 & \displaystyle 0 & \displaystyle 0 & \displaystyle 0 & \displaystyle 0 & \displaystyle 0 & \displaystyle 0 \\
\displaystyle 0 & \displaystyle 0 & \displaystyle 0 & \displaystyle 0 & \displaystyle 0 & \displaystyle 0 & \displaystyle 0 \\
\displaystyle 0 & \displaystyle 0 & \displaystyle 0 & \displaystyle 0 & \displaystyle 0 & \displaystyle 0 & \displaystyle 0
\end{array}\right]\] 
and
\begin{equation}
\label{V10}
V = \left[ \begin{array}{ccccccc}
\displaystyle k_9 Y_{tot} + k_1 & \displaystyle -k_2 & \displaystyle 0 & \displaystyle 0 & \displaystyle 0 & \displaystyle -k_{10} - k_{11} & \displaystyle 0 \\ 
\displaystyle -k_1 & \displaystyle k_2 + k_3 & \displaystyle -k_4 & \displaystyle 0 & \displaystyle -k_8 & \displaystyle 0 & \displaystyle 0 \\ 
\displaystyle 0 & \displaystyle -k_3 & \displaystyle k_{12} Y_{tot} + k_4 + k_5 & \displaystyle 0 & \displaystyle 0 & \displaystyle 0 & \displaystyle -k_{13} - k_{14} \\ 
\displaystyle 0 & \displaystyle 0 & \displaystyle 0 & \displaystyle k_6 X_{tot} & \displaystyle -k_7 & \displaystyle 0 & \displaystyle 0 \\ 
\displaystyle 0 & \displaystyle 0 & \displaystyle 0 & \displaystyle -k_6 X_{tot} & \displaystyle k_7 + k_8 & \displaystyle 0 & \displaystyle 0 \\ 
\displaystyle -k_9 Y_{tot} & \displaystyle 0 & \displaystyle 0 & \displaystyle 0 & \displaystyle 0 & \displaystyle k_{10} + k_{11} & \displaystyle 0 \\ 
\displaystyle 0 & \displaystyle 0 & \displaystyle -k_{12} Y_{tot} & \displaystyle 0 & \displaystyle 0 & \displaystyle 0 & \displaystyle k_{13} + k_{14}
\end{array} \right]
\end{equation}
Permuting the rows and columns to the order $\{X_1, X_2, X_3, X_6, X_7, X_4, X_5\}$ gives the following block triangular form:
\[PVP^T = \left[\begin{array}{ccccccc}
\displaystyle k_9Y_{\text{tot}}  + k_1 & \displaystyle -k_2 & \displaystyle 0 & \displaystyle -k_{10} - k_{11} & \displaystyle 0 & \displaystyle 0 & \displaystyle 0 \\
\displaystyle -k_1 & \displaystyle k_2 + k_3 & \displaystyle -k_4 & \displaystyle 0 & \displaystyle 0 & \displaystyle 0 & \displaystyle -k_8 \\
\displaystyle 0 & \displaystyle -k_3 & \displaystyle k_{12}Y_{\text{tot}}  + k_4 + k_5 & \displaystyle 0 & \displaystyle -k_{13} - k_{14} & \displaystyle 0 & \displaystyle 0 \\
\displaystyle -k_9 Y_{\text{tot}} & \displaystyle 0 & \displaystyle 0 & \displaystyle k_{10} + k_{11} & \displaystyle 0 & \displaystyle 0 & \displaystyle 0 \\
\displaystyle 0 & \displaystyle 0 & \displaystyle -k_{12} Y_{\text{tot}} & \displaystyle 0 & \displaystyle k_{13} + k_{14} & \displaystyle 0 & \displaystyle 0 \\
\displaystyle 0 & \displaystyle 0 & \displaystyle 0 & \displaystyle 0 & \displaystyle 0 & \displaystyle k_6 X_{\text{tot}} & \displaystyle -k_7 \\
\displaystyle 0 & \displaystyle 0 & \displaystyle 0 & \displaystyle 0 & \displaystyle 0 & \displaystyle -k_6 X_{\text{tot}} & \displaystyle k_7 + k_8
\end{array}\right]\]
The blocks are
%\[A_1 = [ \; k_1 E_{tot} \;] \mbox{ and } A_2 = \left[ \begin{array}{cc} k_4 I_{tot} + k_2 + k_3 & -k_5 -k_6 \\ -k_4 I_{tot} & k_5 + k_6 \end{array} \right].\]
\[A_1 = \left[\begin{array}{ccccc}
\displaystyle k_9Y_{\text{tot}}  + k_1 & \displaystyle -k_2 & \displaystyle 0 & \displaystyle -k_{10} - k_{11} & \displaystyle 0 \\
\displaystyle -k_1 & \displaystyle k_2 + k_3 & \displaystyle -k_4 & \displaystyle 0 & \displaystyle 0 \\
\displaystyle 0 & \displaystyle -k_3 & \displaystyle k_{12}Y_{\text{tot}}  + k_4 + k_5 & \displaystyle 0 & \displaystyle -k_{13} - k_{14} \\
\displaystyle -k_9 Y_{\text{tot}} & \displaystyle 0 & \displaystyle 0 & \displaystyle k_{10} + k_{11} & \displaystyle 0 \\
\displaystyle 0 & \displaystyle 0 & \displaystyle -k_{12} Y_{\text{tot}} & \displaystyle 0 & \displaystyle k_{13} + k_{14}
\end{array}\right]\]
and
\[A_2 = \left[\begin{array}{cc}
\displaystyle k_6 X_{\text{tot}} & \displaystyle -k_7 \\
\displaystyle -k_6 X_{\text{tot}} & \displaystyle k_7 + k_8
\end{array}\right]\] 
We have that $\mathbbm{1}^T A_1 = (0,0,k_5,0,0) \geq 0$ and $\mathbbm{1}^T A_2 = (0,k_8) \geq 0$ so that Theorem \ref{thm:2} is satisfied and therefore $V^{-1} \geq 0$. The matrices $V^{-1}$ and $FV^{-1}$ are too lengthy to display. The boundary reproduction number of $\mathbf{x}^*$ for the splitting $\mathcal{F}$ and $\mathcal{V}$ is
\[\mathscr{R}_{\mathbf{x}^*} = Y_{tot} \left( \frac{ ( \left( k_2 k_9 (k_4 + k_5) + k_1 k_3 k_{12} \right) k_{14} + k_2 k_9 k_{13} (k_4 + k_5) ) k_{11} + k_1 k_3 k_{10} k_{12} k_{14}}{k_1 k_3 k_5 (k_{10} + k_{11}) (k_{13} + k_{14})} \right)\]
It follows that $\mathbf{x}^*$ is locally asymptotically stable for every stoichiometric compatibility class satisfying
\[Y_{tot} < \frac{ k_1 k_3 k_5(k_{10} + k_{11}) (k_{13} + k_{14})}{((k_2 k_9 (k_4 + k_5) + k_1 k_3 k_{12}) k_{14} + k_2 k_9 k_{13} (k_4 + k_5)) k_{11} + k_1 k_3 k_{10} k_{12} k_{14}}\]
and unstable for every stoichiometric compatibility class satisfying 
\[Y_{tot} > \frac{ k_1 k_3 k_5(k_{10} + k_{11}) (k_{13} + k_{14})}{((k_2 k_9 (k_4 + k_5) + k_1 k_3 k_{12}) k_{14} + k_2 k_9 k_{13} (k_4 + k_5)) k_{11} + k_1 k_3 k_{10} k_{12} k_{14}}.\]
    
\end{example}

\begin{example}
\label{ex:EnvZOmpR20}
Reconsider the EnvZ-OmpR example from Example \ref{ex:EnvZOmpR19}. We now take the approach of constructing the $\mathcal{X}$-reduced interaction network to identify the appropriate vector $\mathcal{F}$ for constructing the boundary reproduction number $\mathscr{R}_{\mathbf{x}^*}$. The $\mathcal{X}$-reduced interaction network corresponding to \eqref{EnvZOmpR2} is:
\begin{equation}
\label{EnvZOmpR-reduced}
\begin{tikzcd}
X_1 \arrow[rr,yshift=0.5ex,"k_1"] \arrow[dd,xshift=-0.5ex,"k_9"'] & & X_2 \arrow[ll,yshift=-0.5ex,"k_2"] \arrow[rr,yshift=0.5ex,"k_3"] & & X_3 \arrow[ll,yshift=-0.5ex,"k_4"] \arrow[rr,"k_5"] \arrow[dd,xshift=-0.5ex,"k_{12}"'] & & \emptyset \\
\\
X_6 \arrow[uu,xshift=0.5ex,"k_{10}+k_{11}"'] \arrow[ddrr,red,dashed,bend right=15,"k_{11}"'] & & X_5 \arrow[uu,"k_8"'] \arrow[dd,xshift=-0.5ex,"k_7"'] & & X_7 \arrow[uu,xshift=0.5ex,"k_{13}+k_{14}"'] \arrow[ddll,red,dashed,bend left=15,"k_{14}"] & & \\
\\
& & X_4 \arrow[uu,xshift=0.5ex,"k_{6}"'] & & & & 
\end{tikzcd}
\end{equation}
Note that the reactions $k_{11}$ and $k_{14}$—corresponding to the reaction rates $k_{11}x_6$ and $k_{14}x_7$, respectively—appear twice in \eqref{EnvZOmpR-reduced}. 

To apply Theorem \ref{thm:3}, we select the splitting $G_{\mathcal{X}}^{\mathcal{F}}$ (dashed red) and $G_{\mathcal{X}}^{\mathcal{V}}$ (solid black) so that the edges labeled $k_{11}$ and $k_{14}$ do not lie in the same strongly connected component of $G_{\mathcal{X}}^{\mathcal{V}}$. This corresponds to including the $k_{11}x_6$ and $k_{14}x_7$ terms in the $x_4'$ equation of \eqref{EnvZOmpR-de} in $\mathcal{F}$, as in Example \ref{ex:EnvZOmpR19}. This produces the strongly connected components $\{X_4, X_5 \}$, $\{X_1, X_2, X_3, X_6, X_7\}$, and $\{\emptyset\}$, which correspond to the irreducible blocks in \eqref{V10}. Since one instance of $k_{11}$ and $k_{14}$ appears in the strongly connected component $\{X_1, X_2, X_3, X_6, X_7\}$, the requirements of Theorem \ref{thm:3} are satisfied.

This example shows that not every splitting of $\mathcal{F}$ and $\mathcal{V}$ satisfying heuristic \textbf{(H2)} produces $F$ and $V$ matrices that satisfy the requirements of Theorem \ref{thm:3}. For instance, including only the $k_{11}x_6$ term in the $x_4'$ equation of \eqref{EnvZOmpR-de} in $\mathcal{F}$ leads to two $k_{14}$ edges in the same strongly connected component of $G_{\mathcal{X}}^{\mathcal{V}}$. Similarly, including only the $k_{14}x_7$ term leads to two $k_{11}$ edges in the same strongly connected component. In both cases, $V^{-1}$ is not a nonnegative matrix for this choice of $\mathcal{F}$ and $\mathcal{V}$. Consequently, Theorem \ref{thm:main} does not apply, and $\rho(FV^{-1}) = 1$ is not guaranteed to be the correct threshold for stability of the boundary steady state.
\end{example}

\begin{example}
\label{ex:shiu}
Consider the following chemical reaction network introduced in \cite{chavez2003observer} and further studied in \cite{anderson2008global,shiu2010siphons}:
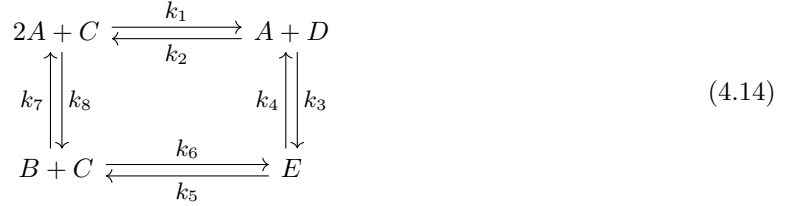
\begin{equation}
\label{shiu1}
\begin{tikzcd}
2A + C \arrow[rr,yshift=0.5ex,"k_1"] \arrow[dd,xshift=0.5ex,"k_8"] & & A+D \arrow[ll,yshift=-0.5ex,"k_2"] \arrow[dd,xshift=0.5ex,"k_3"] \\
& & \\
B+C \arrow[uu,xshift=-0.5ex,"k_7"] \arrow[rr,yshift=0.5ex,"k_6"] & & E \arrow[ll,yshift=-0.5ex,"k_5"] \arrow[uu,xshift=-0.5ex,"k_4"]
\end{tikzcd}
\end{equation}
This yields the following system of differential equations:
\begin{equation}
\label{shiu-de}
\left\{ \; \;
\begin{aligned}
A' &= -(k_1 + 2k_8) A^2 C + (k_2 - k_3) A D + k_4 E + 2k_7 B C \\
B' &= k_5 E - (k_6 + k_7) B C + k_8 A^2 C \\
C' &= -k_1 A^2 C - k_6 B C + k_2 A D + k_5 E \\
D' &= k_1 A^2 C + k_4 E - (k_2 + k_3) A D \\
E' &= k_3 A D + k_6 B C - (k_4 + k_5) E 
\end{aligned}
\right.
\end{equation}
The system \eqref{shiu-de} was shown to be universally persistent—that is, no solution may approach any portion of the boundary for any choice of rate constants—in \cite{anderson2008global}, using methods related to complex-balanced networks. We now consider the boundary reproduction number approach.

The network \eqref{shiu1} has three siphons, which we denote $\mathcal{X}_1 = \{ C, D, E \}$, $\mathcal{X}_2 = \{ A, B, E \}$, and $\mathcal{X}_3 = \{ A, C, E \}$, and two linearly independent conservation laws: $C+D+E = \Lambda_1$ and $A + 2B + D + 2E = \Lambda_2$. We apply the boundary reproduction number approach to the siphons individually. \\

\noindent \textbf{Siphon $\mathcal{X}_1$ (structural persistence)}: Since $A$, $B$, and $C$ are common to $\mathcal{X}_1$ and the first conservation law, $\mathcal{X}_1$ is a noncritical siphon and therefore structurally persistent \cite{angeli2007petri}.\\

\noindent \textbf{Siphon $\mathcal{X}_2$ (universally unstable)}: $\mathcal{X}_2 = \{ A, B, E \}$ is a critical siphon with antisiphon $\mathcal{Y}_2 = \{ C, D \}$. The conservation equation \eqref{conservation2} gives
\[W_{\mathcal{X}_2} = \begin{bmatrix} 0 & 0 & 1 \\ 1 & 2 & 2 \end{bmatrix}, \quad W_{\mathcal{Y}_2} = \begin{bmatrix} 1 & 1 \\ 0 & 1 \end{bmatrix}, \quad W_{\mathcal{Y}_2}^{-1} = \begin{bmatrix} 1 & -1 \\ 0 & 1 \end{bmatrix}, \quad \text{and} \quad \mathbf{\Lambda} = \begin{bmatrix} \Lambda_1 \\ \Lambda_2 \end{bmatrix}.\]
We have $|\mathcal{Y}_2| = n - s = 2$ and $W_{\mathcal{Y}_2}$ is invertible. Consequently, every $\mathcal{X}_2$-free boundary steady state has the form $\mathbf{x}_2^* = W_{\mathcal{Y}_2}^{-1} \mathbf{\Lambda} = (0, 0, \Lambda_1 - \Lambda_2, \Lambda_2, 0)$. Note that $\mathbf{x}_2^* \geq 0$ if $\Lambda_1 > \Lambda_2$.

We select the positive terms $k_2 AD + 2k_7 BC + k_4 E$ in the $A'$ equation to include in $\mathcal{F}$ and exclude the rest. We now take the Jacobian of $\mathcal{F}$ and $\mathcal{V}$ with respect to the species in $\mathcal{X}_2$ and evaluate at $\mathbf{x}_2^* = (0, 0, \Lambda_1 - \Lambda_2, \Lambda_2, 0) \in E_{\mathcal{X}_2}$. This yields 
\[F = \begin{bmatrix}
k_2 \Lambda_2 & 2k_7(\Lambda_1 - \Lambda_2) & k_4 \\
0 & 0 & 0 \\
0 & 0 & 0
\end{bmatrix}, \quad V = \begin{bmatrix}
k_3 \Lambda_2 & 0 & 0 \\
0 & (k_6 + k_7)(\Lambda_1 - \Lambda_2) & -k_5 \\
- k_3 \Lambda_2 & -k_6(\Lambda_1 - \Lambda_2) & k_4 + k_5
\end{bmatrix}\]
so that
\[
V^{-1} = \begin{bmatrix}
\displaystyle \frac{1}{k_3 \Lambda_2} & 0 & 0 \\
\displaystyle \frac{k_5}{(k_4 k_6 + k_4 k_7 + k_5 k_7)(\Lambda_1 - \Lambda_2)} & 
\displaystyle \frac{k_4 + k_5}{(k_4 k_6 + k_4 k_7 + k_5 k_7)(\Lambda_1 - \Lambda_2)} & 
\displaystyle \frac{k_5}{(k_4 k_6 + k_4 k_7 + k_5 k_7)(\Lambda_1 - \Lambda_2)} \\
\displaystyle \frac{k_6 + k_7}{k_4 k_6 + k_4 k_7 + k_5 k_7} & 
\displaystyle \frac{k_6}{k_4 k_6 + k_4 k_7 + k_5 k_7} & 
\displaystyle \frac{k_6 + k_7}{k_4 k_6 + k_4 k_7 + k_5 k_7}
\end{bmatrix}
\]
and
\[
FV^{-1} = \begin{bmatrix}
\displaystyle \frac{(k_6 + k_7)(k_2 + k_3)k_4 + k_5 k_7(k_2 + 2k_3)}{k_3((k_6 + k_7)k_4 + k_5 k_7)} &
\displaystyle \frac{(k_6 + 2k_7)k_4 + 2k_5 k_7}{(k_6 + k_7)k_4 + k_5 k_7} &
\displaystyle \frac{(k_6 + k_7)k_4 + 2k_5 k_7}{(k_6 + k_7)k_4 + k_5 k_7} \\
0 & 0 & 0 \\
0 & 0 & 0
\end{bmatrix}.
\]
Provided $\Lambda_1 > \Lambda_2$, we have $F \geq 0$, $V^{-1} \geq 0$, and $V$ is a $Z$-matrix, so that Definition \ref{def:bpn} and Theorem \ref{thm:main} apply. It follows that the boundary reproduction number of $\mathbf{x}_2^*$ for this splitting is given by
\[\mathscr{R}_{\mathbf{x}_2^*} = \frac{k_4(k_2 + k_3)(k_6 + k_7) + k_5 k_7 (k_2 + 2k_3)}{k_3 (k_4 (k_6 + k_7) + k_5 k_7)}.\]
The condition $\mathscr{R}_{\mathbf{x}_2^*} > 1$ is equivalent to $k_2 k_4 (k_6 + k_7) + k_5 k_7 (k_2 + k_3) > 0$, which always holds. Theorem \ref{thm:main} therefore guarantees that every $\mathcal{X}_2$-free boundary steady state is unstable, regardless of the rate constant values or stoichiometric compatibility class.\\

\noindent \textbf{Siphon $\mathcal{X}_3$ (universally unstable)}: $\mathcal{X}_3 = \{ A, C, E \}$ is a critical siphon with antisiphon $\mathcal{Y}_3 = \{ B, D \}$. The conservation equation \eqref{conservation2} gives
\[
W_{\mathcal{X}_3} = \begin{bmatrix} 0 & 1 & 1 \\ 1 & 0 & 2 \end{bmatrix}, \quad W_{\mathcal{Y}_3} = \begin{bmatrix} 0 & 1 \\ 2 & 1 \end{bmatrix}, \quad W_{\mathcal{Y}_3}^{-1} = \frac{1}{2} \begin{bmatrix} -1 & 1 \\ 2 & 0 \end{bmatrix}, \quad \text{and} \quad \mathbf{\Lambda} = \begin{bmatrix} \Lambda_1 \\ \Lambda_2 \end{bmatrix}.
\]
We have $|\mathcal{Y}_3| = n-s = 2$ and $W_{\mathcal{Y}_3}$ is invertible. Consequently, every $\mathcal{X}_3$-free boundary steady state has the form $\mathbf{x}_3^* = W_{\mathcal{Y}_3}^{-1} \mathbf{\Lambda} = (0, \frac{1}{2}(\Lambda_2 - \Lambda_1), 0, \Lambda_1, 0)$. Note that $\mathbf{x}_3^* \geq 0$ if $\Lambda_2 > \Lambda_1$.

We again select the positive terms $k_2 AD + 2k_7 BC + k_4 E$ in the $A'$ equation to include in $\mathcal{F}$ and exclude the rest. We then take the Jacobian of $\mathcal{F}$ and $\mathcal{V}$ with respect to the species $\mathcal{X}_3$ and evaluate at an arbitrary point $\mathbf{x}_3^* = (0, \frac{1}{2}(\Lambda_2 - \Lambda_1), 0, \Lambda_1, 0) \in E_{\mathcal{X}_3}$. This yields
\[
F = \begin{bmatrix}
 k_{2}\Lambda_1 & k_7(\Lambda_2 - \Lambda_1) & k_4 \\
0 & 0 & 0 \\
0 & 0 & 0
\end{bmatrix}, \quad
V = \begin{bmatrix}
k_3 \Lambda_1 & 0 & 0 \\
- k_2 \Lambda_1 & \displaystyle \frac{k_6}{2} (\Lambda_2 - \Lambda_1) & -k_5 \\
- k_3 \Lambda_1 & \displaystyle -\frac{k_6}{2} (\Lambda_2 - \Lambda_1) & k_4 + k_5
\end{bmatrix},
\]
so that
\[
V^{-1} = \begin{bmatrix}
\displaystyle \frac{1}{k_3 \Lambda_1} & 0 & 0 \\
\displaystyle \frac{2(k_2 k_4 + k_2 k_5 + k_3 k_5)}{k_3 k_4 k_6 (\Lambda_2 - \Lambda_1)} & \displaystyle \frac{2(k_4 + k_5)}{k_4 k_6 (\Lambda_2 - \Lambda_1)} & \displaystyle \frac{2k_5}{k_4 k_6 (\Lambda_2 - \Lambda_1)} \\
\displaystyle \frac{k_2 + k_3}{k_3 k_4} & \displaystyle \frac{1}{k_4} & \displaystyle \frac{1}{k_4}
\end{bmatrix},
\]
and
\[
F V^{-1} = \begin{bmatrix}
\displaystyle \frac{2(k_2 k_4 (k_6 + k_7) + k_5 k_7 (k_2 + k_3)) + k_3 k_4 k_6}{k_3 k_4 k_6} & \displaystyle \frac{k_4 (k_6 + 2 k_7) + 2 k_5 k_7}{k_4 k_6} & \displaystyle \frac{k_4 k_6 + 2 k_5 k_7}{k_4 k_6} \\
0 & 0 & 0 \\
0 & 0 & 0
\end{bmatrix}.
\]
Provided $\Lambda_2 > \Lambda_1$, we have $F \geq 0$, $V^{-1} \geq 0$, and $V$ is a $Z$-matrix, so Definition \ref{def:bpn} and Theorem \ref{thm:main} apply. It follows that the boundary reproduction number of $\mathbf{x}_3^*$ for this splitting is
\[
\mathscr{R}_{\mathbf{x}_3^*} = \frac{2(k_2 k_4 (k_6 + k_7) + k_5 k_7 (k_2 + k_3)) + k_3 k_4 k_6}{k_3 k_4 k_6}.
\]
The condition $\mathscr{R}_{\mathbf{x}_3^*} > 1$ is equivalent to $k_2 k_4 (k_6 + k_7) + k_5 k_7 (k_2 + k_3) > 0$, which always holds. Theorem \ref{thm:main} therefore guarantees that every $\mathcal{X}_3$-free boundary steady state is unstable, regardless of rate constant values or stoichiometric compatibility class.\\

\noindent \textbf{Summary}: We have used the boundary reproduction number method to show that every boundary steady state is unstable, regardless of rate constant values and stoichiometric compatibility classes. However, the method cannot guarantee persistence with respect to the boundary, as instability does not rule out saddle-like behavior (i.e., the presence of stable and unstable manifolds). Establishing conditions under which the boundary reproduction number approach can be extended from guaranteeing local stability or instability to guaranteeing broader persistence is a topic for future work.

\end{example}

\section{Conclusions}
\label{sec:conclusions}

We have introduced the boundary reproduction number of a boundary steady state as a method for establishing the stability and instability of boundary steady states in biochemical reaction systems. Our results extend applications of the next generation matrix method (see \cite{VANDENDRIESSCHE2002,van2008further}) and formalize the work of \cite{Avram2025StabilityIR} within the language and theory of chemical reaction networks. When successful, the boundary reproduction number approach significantly reduces the computational time required to determine thresholds between stability and instability for boundary steady states.

We also presented matrix-based (Theorem \ref{thm:2}) and network-based (Theorem \ref{thm:3}) conditions under which the boundary reproduction number approach is successful. These results further illuminate how the selection of terms included in the vector $\mathcal{F}$ influences the computational complexity of computing the boundary reproduction number. We introduce heuristics for determining $\mathcal{F}$, which can be further refined by considering the network structure of the $\mathcal{X}$-reduced network $G_{\mathcal{X}}$. In many cases, we find that $\mathcal{F}$ can be chosen so that the resulting next generation matrix has rank one, in which case the boundary reproduction number can be directly read from the matrix. These results and heuristics provide insight for both biochemical reaction systems and models of infectious disease spread.

We pose the following open questions for future research:

\begin{enumerate}
\item \textbf{Extension from Instability to Persistence:} Theorem \ref{thm:main} guarantees that when $\mathscr{R}_{\mathbf{x}^*}$ exceeds one, the steady state $\mathbf{x}^*$ transitions from stable to unstable. Consequently, every neighborhood of $\mathbf{x}^*$ contains at least one trajectory that leaves the neighborhood; however, since the steady state may still exhibit saddle-like behavior, it is not guaranteed that \emph{every} trajectory leaves the neighborhood, as is required in the definition of persistence commonly studied in chemical reaction network theory. We therefore pose the question: What conditions on a chemical reaction system \eqref{de} guarantee the persistence of \eqref{de} with respect to a boundary face when $\mathscr{R}_{\mathbf{x}^*}>1$?

\item \textbf{Rank-One Next Generation Matrix:} Our examples illustrate that the choice of $\mathcal{F}$ and $\mathcal{V}$ matters, as some selections may satisfy the assumptions of Theorem \ref{thm:main}, while others may not. We have also seen that including fewer terms in $\mathcal{F}$ often simplifies computation, with the ideal scenario being that $F$, and hence the next generation matrix $FV^{-1}$, has rank one. Neither Theorem \ref{thm:main} nor our heuristic \textbf{(H1)}–\textbf{(H3)} guarantees the existence of a rank-one $F$ nor does it provide a method for identifying it short of exhaustive search. We therefore pose the question: Under what conditions does a \emph{rank-one} next generation matrix satisfying Theorem \ref{thm:main} exist, and how can $\mathcal{F}$ and $\mathcal{V}$ be selected to produce it?
\end{enumerate}

\paragraph{Funding} MDJ is supported by National Science Foundation Grant No. DMS-2213390.

%\bibliographystyle{plain}
%\bibliography{mybib}

\appendix
\section{Proof of Theorem \ref{thm:main}}
\label{app:a}

We introduce the following matrix structures (see \cite{VANDENDRIESSCHE2002,van2008further}) and result (see Theorem 6.2.3 in \cite{berman1994nonnegative}, also stated as a lemma in \cite{van2008further}).

\begin{definition}
\label{matrices1}
A matrix $A \in \mathbb{R}^{n \times n}$ is an \textbf{$M$-matrix} if $A$ is a $Z$-matrix with $s \geq \rho(B)$, where $\rho(B)$ is the spectral radius of $B$ (equivalently, $A$ is a matrix with $a_{ij} \leq 0$ for $i \not= j$ and all eigenvalues having positive real part).
\end{definition}

\begin{lemma}
\label{lemma3}
If $A \in \mathbb{R}^{n \times n}$ is a $Z$-matrix, then $A^{-1} \geq 0$ if and only if $A$ is a nonsingular $M$-matrix.
\end{lemma}

We are now prepared to prove Theorem \ref{thm:main}.

\begin{proof}[Proof of Theorem \ref{thm:main}]
Consider a chemical reaction network $(\mathcal{S}, \mathcal{R})$ and an $\mathcal{X}$-free boundary steady state $\mathbf{x}^* = ( \mathbf{0}, \tilde{\mathbf{y}}^*) \in E_{\mathcal{X}}$. Define $\mathbf{f}(\tilde{\mathbf{x}},\tilde{\mathbf{y}})$ and $\mathbf{g}(\tilde{\mathbf{x}},\tilde{\mathbf{y}})$ as in \eqref{de2}, and $F$ and $V$ as in \eqref{FV}. The Jacobian of \eqref{de} evaluated at $\mathbf{x}^* = (\mathbf{0}, \tilde{\mathbf{y}}^*)$ is given by
\begin{equation}
\label{eq121}
J = \begin{bmatrix}
\displaystyle \frac{\partial \mathbf{f}}{\partial \tilde{\mathbf{x}}}(\mathbf{0}, \tilde{\mathbf{y}}^*) & \displaystyle \frac{\partial \mathbf{f}}{\partial \tilde{\mathbf{y}}}(\mathbf{0}, \tilde{\mathbf{y}}^*) \\
\displaystyle \frac{\partial \mathbf{g}}{\partial \tilde{\mathbf{x}}}(\mathbf{0}, \tilde{\mathbf{y}}^*) & \displaystyle \frac{\partial \mathbf{g}}{\partial \tilde{\mathbf{y}}}(\mathbf{0}, \tilde{\mathbf{y}}^*) \\
\end{bmatrix} = \begin{bmatrix}
\displaystyle \frac{\partial \mathcal{F}}{\partial \tilde{\mathbf{x}}}(\mathbf{0}, \tilde{\mathbf{y}}^*) - \frac{\partial \mathcal{V}}{\partial \tilde{\mathbf{x}}}(\mathbf{0}, \tilde{\mathbf{y}}^*) & 0 \\
\displaystyle \frac{\partial \mathbf{g}}{\partial \tilde{\mathbf{x}}}(\mathbf{0}, \tilde{\mathbf{y}}^*) & \displaystyle \frac{\partial \mathbf{g}}{\partial \tilde{\mathbf{y}}}(\mathbf{0}, \tilde{\mathbf{y}}^*) \\
\end{bmatrix} = \begin{bmatrix}
F-V & 0 \\
J_{12} & J_{22}
\end{bmatrix}
\end{equation}
where $\displaystyle \frac{\partial \mathbf{f}}{\partial \tilde{\mathbf{y}}}(\mathbf{0}, \tilde{\mathbf{y}}^*) = 0$ because $f(\mathbf{0},\tilde{\mathbf{y}}) = 0$ for all $\tilde{\mathbf{y}}$, by properties of siphons (specifically, that no species $X_i \in \mathcal{X}$ may be created without being present in the reactant). By the block structure of \eqref{eq121}, to determine the stability or instability of $\mathbf{x}^*$, it is sufficient to consider the eigenvalues and eigenvectors of $F-V$ and $J_{22}$.

Consider $F-V$. Since $V^{-1} \geq 0$ and $V$ is a $Z$-matrix by assumption, Lemma \ref{lemma3} implies that $V$ is a nonsingular $M$-matrix. Suppose that all the eigenvalues of $V-F$ have positive real part. Since $F \geq 0$ by construction and $V$ is a $Z$-matrix, it follows that $V-F$ is a $Z$-matrix. Thus, $V-F$ is a nonsingular $M$-matrix, and Lemma \ref{lemma3} gives $(V-F)^{-1} \geq 0$. Therefore,
\begin{equation}
\label{eq:200}
0 \leq I + F(V-F)^{-1} = (V-F+F)(V-F)^{-1} = V(V-F)^{-1} = [(V-F)V^{-1}]^{-1} = (I-FV^{-1})^{-1}.
\end{equation}
Since $F \geq 0$ and $V^{-1} \geq 0$, we have that $I-FV^{-1}$ is a $Z$-matrix and hence a nonsingular $M$-matrix by Lemma \ref{lemma3}. Consequently, the eigenvalues of $I-FV^{-1}$ have positive real part. Since $\lambda$ is an eigenvalue of $A$ if and only if $1-\lambda$ is an eigenvalue of $I-A$, it follows that $\mathscr{R}_{\mathbf{x}^*} = \rho(FV^{-1}) < 1$.

Now suppose $\rho(FV^{-1}) < 1$. Then the eigenvalues of $I-FV^{-1}$ have positive real part. Since $I-FV^{-1}$ is a $Z$-matrix, it is a nonsingular $M$-matrix, so $(I-FV^{-1})^{-1} \geq 0$. Since $V^{-1} \geq 0$, this implies
\begin{equation}
\label{eq:201}
0 \leq V^{-1}(I-FV^{-1})^{-1} = [(I-FV^{-1})V]^{-1} = (V-F)^{-1}.
\end{equation}
Since $V-F$ is a $Z$-matrix, it follows by Lemma \ref{lemma3} that $V-F$ is a nonsingular $M$-matrix. Hence, the eigenvalues of $V-F$ have positive real part. It follows that $J$ has eigenvalues with negative real part if and only if $\mathscr{R}_{\mathbf{x}^*} = \rho(FV^{-1}) < 1$. 

To show that $J$ has an eigenvalue with positive real part if and only if $\mathscr{R}_{\mathbf{x}^*} > 1$, we apply a continuity argument. If the eigenvalues of $V-F$ have nonnegative real part, then $V-F+\epsilon I$ is a nonsingular $M$-matrix for any $\epsilon > 0$. It follows by \eqref{eq:200} that $\mathscr{R}_{\mathbf{x}^*} = \rho(F(V+\epsilon I)^{-1}) < 1$. Taking $\epsilon \to 0$ gives $\rho(FV^{-1}) \leq 1$ by continuity of eigenvalues. Conversely, if $\rho(FV^{-1}) \leq 1$, then $(1+\epsilon)I-FV^{-1}$ is a nonsingular $M$-matrix for any $\epsilon > 0$. It follows by \eqref{eq:201} that the eigenvalues of $(1+\epsilon)V-F$ have positive real part. Taking $\epsilon \to 0$ gives that the eigenvalues of $V-F$ have nonnegative real part by continuity of eigenvalues. Therefore, $\mathscr{R}_{\mathbf{x}^*} \leq 1$ if and only if the eigenvalues of $J$ have nonpositive real part. The contrapositive yields that $\mathscr{R}_{\mathbf{x}^*} > 1$ if and only if $J$ has an eigenvalue with positive real part.

Now consider $J_{22}$. If there are no conservation laws, then all eigenvalues of $J_{22}$ have negative real part by assumption. Therefore, by the block structure of \eqref{eq121}, if $\mathscr{R}_{\mathbf{x}^*}<1$, then all eigenvalues of $J$ have negative real part, and stability follows. Conversely, if $\mathscr{R}_{\mathbf{x}^*} > 1$, then $J$ has at least one eigenvalue with positive real part, and instability follows.

If the number of conservation laws equals the number of elements in the antisiphon (i.e., $|\mathcal{Y}| = n - s$) and $W_{\mathcal{Y}}$ is invertible, then \eqref{de2} reduces to the system
\[
\left\{ \; \; \;
    \begin{aligned}
 \frac{d\tilde{\mathbf{x}}}{dt} & = \mathbf{f}(\tilde{\mathbf{x}},\tilde{\mathbf{y}}) \\
    \tilde{\mathbf{y}} & = W_{\mathcal{Y}}^{-1} \left( \mathbf{\Lambda} - W_{\mathcal{X}} \tilde{\mathbf{x}} \right).
    \end{aligned}
\right.
\]
Since $W_{\mathcal{Y}}$ is invertible and $\mathbf{\Lambda}$ uniquely determines the stoichiometric compatibility class, there is a unique $\mathcal{X}$-free boundary steady state $\mathbf{x}^* = (\mathbf{0}, \tilde{\mathbf{y}}^*) = (\mathbf{0}, W_{\mathcal{Y}}^{-1} \mathbf{\Lambda})$ in each stoichiometric compatibility class. The Jacobian evaluated at $\mathbf{x}^*$ is
\[
J(\mathbf{0},W_{\mathcal{Y}}^{-1} \mathbf{\Lambda}) = \frac{\partial \mathbf{f}}{\partial \tilde{\mathbf{x}}}(\mathbf{0},W_{\mathcal{Y}}^{-1} \mathbf{\Lambda}) + \frac{\partial \mathbf{f}}{\partial \tilde{\mathbf{y}}}(\mathbf{0},W_{\mathcal{Y}}^{-1} \mathbf{\Lambda}) \frac{\partial \tilde{\mathbf{y}}}{\partial \tilde{\mathbf{x}}} = F - V
\]
since $\displaystyle \frac{\partial \mathbf{f}}{\partial \tilde{\mathbf{y}}}(\mathbf{0},W_{\mathcal{Y}}^{-1} \mathbf{\Lambda}) = \mathbf{0}$, again because $\mathcal{X}$ is a siphon. If $\mathscr{R}_{\mathbf{x}^*} < 1$, then all eigenvalues of $J$ have negative real part, so that $\tilde{\mathbf{x}}(t) \to \mathbf{0}$ and $\tilde{\mathbf{y}}(t) \to W_{\mathcal{Y}}^{-1} \mathbf{\Lambda}$ by the second equation. Consequently, $\mathbf{x}^* = (\mathbf{0},W_{\mathcal{Y}}^{-1} \mathbf{\Lambda})$ is locally asymptotically stable. If $\mathscr{R}_{\mathbf{x}^*} > 1$, then the steady state is unstable.
\end{proof}

\section{Proof of Theorem \ref{thm:2}}
\label{app:b}

To prove Theorem \ref{thm:2}, we connect the matrices $V$ and $V^{-1}$ to the transition probability rate and fundamental matrices from the theory of continuous-time Markov chains. We introduce the following definitions (see \cite{norris1998markov}).

\begin{definition}
\label{generator}
Let \( \{ X(t), t \geq 0 \} \) be a continuous-time Markov chain with finite state space \( S = \{ 1, \ldots, n \} \). We have the following:
\begin{enumerate}
\item
A state $i \in S$ is called \textbf{recurrent} if a sequence of transitions from $i$ to $j$ where $j \in S$ implies that there is a sequence of transitions from $j$ to $i$. The set of \textbf{recurrent states} is denoted \( R \subseteq S \).
\item 
A state $i \in S$ is called \textbf{transient} if there is a sequence of transitions from $i$ to $j$ where $j \in S$ such that there is no sequence of transitions from $j$ to $i$. The set of \textbf{transient states} is denoted \( T \subseteq S \).
\item
The \textbf{generator matrix} \( Q = (q_{ij}) \), where \( q_{ij} \geq 0 \) for $i \not= j$ and \( q_{ii} = -\sum_{j \neq i} q_{ij} \), represents the transition rates from \( i \) to \( j \). The generator matrix \( Q \) can be partitioned as:
\begin{equation}
\label{Q}
Q =
\begin{bmatrix}
Q_{TT} & Q_{TR} \\
0 & Q_{RR}
\end{bmatrix}
\end{equation}
where \( Q_{TT} \) corresponds to transitions between transient states, \( Q_{TR} \) to transitions from transient to recurrent states, and \( Q_{RR} \) to transitions between recurrent states.
\item
The \textbf{fundamental matrix} \( N \) for the transient states is given by:
\[
N = (-Q_{TT})^{-1}
\]
where the entry \( n_{ij} \geq 0 \) corresponds to the expected time spent in state $j$ when starting in state $i$ before leaving the transient component.
\end{enumerate}
\end{definition}

We now make the correspondence between the matrices $A_i$ in \eqref{blocktriangular} and the negative transpose of the generator matrix $Q$ of a finite-state Markov chain restricted to the transient states (i.e., $(-Q_{TT})^T$). Specifically, we note that $(-Q_{TT})^T$ is a $Z$-matrix by construction. Furthermore, since $Q \mathbbm{1} = \mathbf{0}$ and there is at least one transition from $T$ to $R$ so that $Q_{TR} \geq 0$, it follows from \eqref{Q} that $Q_{TT} \mathbbm{1} + Q_{TR} \mathbbm{1} = 0$. Hence, $Q_{TT} \mathbbm{1} = - Q_{TR} \mathbbm{1} \leq 0$, and therefore $\mathbbm{1}^T (-Q_{TT})^T \geq 0$.

Consequently, the matrices $A_i$ in \eqref{blocktriangular} satisfying the conditions of Corollary \ref{thm:2} (i.e., $A_i$ is a $Z$-matrix and $\mathbbm{1}^T A_i \geq 0$) have the same structure as the negative transpose of the generator matrix $Q$ of a finite-state Markov chain restricted to the transient states (i.e., $A_i = (-Q_{TT})^T$ for some finite-state Markov chain). The following result follows from Theorem 1 of \cite{plemmons1977m} and \cite{kemeny1976absorbing}.

\begin{lemma}
\label{lemma11}
If $A \in \mathbb{R}^{n \times n}$ is a $Z$-matrix satisfying $\mathbbm{1}^T A \geq 0$, then $A$ is a nonsingular $M$-matrix (i.e., $A^{-1} \geq 0$).
\end{lemma}

\begin{proof}
If $A \in \mathbb{R}^{n \times n}$ is a $Z$-matrix and $\mathbbm{1}^T A \geq 0$, then $A = (-Q_{TT})^T$ where $-Q_{TT}$ is the negative transpose of the generator of a Markov chain restricted to the transient states. From Theorem 1 of \cite{plemmons1977m}, we have $(-Q_{TT})^{-1} \geq 0$, so that $A^{-1} = ((-Q_{TT})^{-1})^T \geq 0$. Since $A$ is a $Z$-matrix, it follows from Lemma \ref{lemma3} that $A$ is an $M$-matrix, and we are done.
\end{proof}

For models of infectious disease spread, the entries $[V^{-1}]_{ij}$ have been interpreted as the expected ``dwell time'' an individual spends in infectious state $i$ when starting in infectious state $j$ before leaving the infectious states (i.e., before recovering) \cite{VANDENDRIESSCHE2002,van2017basic}. The product $FV^{-1}$ therefore corresponds to the combined ``force of infection'' ($F$) and time spent in each infectious state ($V^{-1}$). For biochemical reaction networks, we are not tracking individuals transitioning through infectious states, so our interpretation of the matrices $A_i^{-1}$ will not be as intuitive.

We are now prepared to prove Theorem \ref{thm:2}.

\begin{proof}[Proof of Theorem \ref{thm:2}]
Suppose $V$ is block triangularizable with block triangular form \eqref{blocktriangular} where $J_{ij} \leq 0$ and, for each $A_i$, $i = 1, \ldots, n$, $A_i$ is a $Z$-matrix and $\mathbbm{1}^T A_i \geq 0$. Each $A_i$ is an $M$-matrix by Lemma \ref{lemma11}, and consequently every eigenvalue of every $A_i$ has positive real part. Since permutation is a similarity transformation, $V$ has the same eigenvalues as its block triangular form. It follows from the block structure of \eqref{blocktriangular} that $\sigma(V) = \bigcup_{i=1}^n \sigma(A_i)$ and $\rho(V) = \max_{i = 1, \ldots, n} \{ \rho(A_i) \}$. Since the $A_i$ are $Z$-matrices and $J_{ij} \leq 0$, we have that $V$ is a $Z$-matrix, and thus $V$ is a nonsingular $M$-matrix. It follows from Lemma \ref{lemma3} that $V^{-1} \geq 0$, so Theorem \ref{thm:main} applies, and we are done.
\end{proof}

\section{Proof of Theorem \ref{thm:3}}
\label{app:c}

We now prove Theorem \ref{thm:3}.

\begin{proof}[Proof of Theorem \ref{thm:3}]
Consider the $\mathcal{X}$-reduced network $G_{\mathcal{X}} = (V_{\mathcal{X}},E_{\mathcal{X}} )$ of a chemical reaction network $(\mathcal{S},\mathcal{R})$ with a critical siphon $\mathcal{X}$.  Suppose that kinetic assumptions \textbf{(A1)}-\textbf{(A4)} are satisfied. Suppose further that there is a splitting of $G_{\mathcal{X}} = G^{\mathcal{F}}_{\mathcal{X}} \cup G^{\mathcal{V}}_{\mathcal{X}}$ where $G^{\mathcal{F}}_{\mathcal{X}} = (V_{\mathcal{X}},E^{\mathcal{F}}_{\mathcal{X}})$ and $G^{\mathcal{V}}_{\mathcal{X}} = (V_{\mathcal{X}},E^{\mathcal{V}}_{\mathcal{X}})$ satisfy Conditions 1-4 of Theorem \ref{thm:3}. We show that we can construct vectors $\mathcal{F}$ and $\mathcal{V}$ so that the matrices $F$ and $V$ given by \eqref{FV} satisfy $F \geq 0$, $V$ is a $Z$-matrix, and $V$ is block triangularizable with block triangular form \eqref{blocktriangular} where $J_{ij} \leq 0$ and, for each $A_i$, $i = 1, \ldots, n$, $A_i$ is a $Z$-matrix and $\mathbbm{1}^T A_i \geq 0$. That is, we show that $F$ and $V$ satisfy the assumptions of Theorem \ref{thm:2} so that $V^{-1} \geq 0$ follows. %Consequently, the splitting of $E_{\mathcal{X}}$ into $E_{\mathcal{X}}^{\mathcal{F}}$ and $E_{\mathcal{X}}^{\mathcal{V}}$ described by Theorem \ref{thm:3} completely determines the splitting of the Jacobian of $\mathbf{f}(\tilde{\mathbf{x}},\tilde{\mathbf{y}})$ into $F$ and $V$ as defined in \eqref{FV}. Specifically, to $\mathcal{F}$ we add all reaction rates in \eqref{de2} associated with targets of edges in $E^{\mathcal{F}}_{\mathcal{X}}$ and to $\mathcal{V}$ we add all reaction rates associated with targets of edges in $E^{\mathcal{V}}_{\mathcal{X}}$. 
We now assess the required choices for $\mathcal{F}$ and $\mathcal{V}$.

To $\mathcal{F}$ we add all reaction rates in \eqref{de2} associated with targets of edges in $E^{\mathcal{F}}_{\mathcal{X}}$. Since $E^{\mathcal{F}}_{\mathcal{X}} \not= \emptyset$ by Condition 1 of Theorem \ref{thm:3}, we have $\mathcal{F} \geq 0$. By kinetic assumption \textbf{(A4)}, it follows that $F \geq 0$.

To $-\mathcal{V}$ we add all reaction rates associated with edges in $E^{\mathcal{V}}_{\mathcal{X}}$, which correspond to the remaining terms in \eqref{de2}. Condition 3 of Theorem \ref{thm:3} guarantees that all terms corresponding to autocatalytic reactions (i.e. those which increase the amount of $X_i$ as $X_i$ is increased) are in $\mathcal{F}$ rather than $-\mathcal{V}.$ All remaining terms correspond to edges of the form $X_i \to X_j$, for $X_i \not= X_j$, or $X_i \to \emptyset$. Condition 4(a) of Theorem \ref{thm:3}, together with kinetic assumption \textbf{(A4)}, guarantees that each reaction rate $R(\tilde{\mathbf{x}},\tilde{\mathbf{y}})$ contributes a single term $\displaystyle \frac{\partial R(\mathbf{0},\tilde{\mathbf{y}}^*)}{\partial x_i}>0$ to the $(i,i)$ entry of $V$ and as many terms $\displaystyle \frac{\partial R(\mathbf{0},\tilde{\mathbf{y}}^*)}{\partial x_i}<0$ to the off-diagonal $(i,j)$ entries of $V$ as there are edges $X_i \to X_j$ associated with the reaction rate $R(\tilde{\mathbf{x}},\tilde{\mathbf{y}})$ in $E_{\mathcal{X}}^{\mathcal{V}}$. By Condition 4(b) of Theorem \ref{thm:3}, the entries of $V$ are non-positive on the off-diagonal so that $V$ is a $Z$-matrix.

%Every unique reaction rate $R(\tilde{\mathbf{x}},\tilde{\mathbf{y}})$

%that every edge in $E_{\mathcal{X}}$ corresponds to a single nonzero term in the off-diagonal portion of the Jacobian of $\mathbf{f}(\tilde{\mathbf{x}},\tilde{\mathbf{y}})$ evaluated at the boundary steady state, $J(\mathbf{0},\tilde{\mathbf{y}}^*)$. Specifically, if there is an edge from $X_i$ to $X_j$ in $E_{\mathcal{X}}$, then there is a term $\frac{\partial R(\mathbf{0},\tilde{\mathbf{y}}^*)}{\partial x_i}$ in the $(i,j)$ entry of $J(\mathbf{0},\tilde{\mathbf{y}}^*)$. The strongly connected components of $G^{\mathcal{V}}_{\mathcal{X}}$ correspond to blocks $A_i$ in the block triangular structure of $V$ \eqref{blocktriangular}.
%\end{enumerate}

%Consequently, the only terms on the diagonal of $-V$ correspond to terms which decrease $X_i$ and these terms cannot vanish by kinetic assumption \textbf{(A4)}. Specifically, there is a term $-\frac{\partial R(\mathbf{0},\tilde{\mathbf{y}}^*)}{\partial x_i}$ in the $(i,i)$ entry of $J(\mathbf{0},\tilde{\mathbf{y}}^*)$ for every unique reaction rate $R(\tilde{\mathbf{x}},\tilde{\mathbf{y}})$ from $X_i$. By Condition 2 of Theorem \ref{thm:3}, the unique terminal strongly connected component is $\{ \emptyset\}$ so that every species $X_i \in \mathcal{X}$ is the source for at least one edge in $E_{\mathcal{X}}^{\mathcal{V}}$. It follows that the diagonal entries of $\mathcal{V}$ are strictly positive.

We now consider the strongly connected components of $G_{\mathcal{X}}^{\mathcal{V}}$. After permutation, these correspond to the irreducible diagonal blocks $A_i$ in \eqref{blocktriangular}. By Condition 4(b) of Theorem \ref{thm:3}, every column of $A_i$ can have at most one non-zero off-diagonal term corresponding to a given reaction. If there is a non-zero off-diagonal term corresponding to this reaction rate, it is the same term as appears on the diagonal but opposite sign and the only addition edges correspond to $X_i \to \emptyset$ which add a positive term to the diagonal. It follows that $\mathbbm{1}^T A_i$ is nonnegative in every entry. Furthermore, since any strongly connected component containing an $X_i$ is non-terminal, there must be either an edge $X_i \to \emptyset$ or an edge $X_i \to X_j$ where $X_j$ belongs to another strongly connected component. Since this term contributes no term to the off-diagonal entries of $A_i$, we must have that there is at least one column of $A_i$ for which there is a strictly positive column sum. It follows that, for every $A_i$, $\mathbbm{1}^T A_i \geq 0$. Consequently, $V^{-1} \geq 0$ by Theorem \ref{thm:2}, and we are done.
%Condition 4(a) of Theorem \ref{thm:3} implies that the off-diagonal entries of $V$ corresponding to a strongly connected component are nonpositive. Condition 4(b) of Theorem \ref{thm:3} implies that each column of $A_i$ (corresponding to a strongly connected component) contains at most one off-diagonal entry corresponding to each reaction rate. Therefore, we have that $\mathbbm{1}^T A_i \geq 0$ for every nonterminal strongly connected components $A_i$. Conditions 2 and 4 of Theorem \ref{thm:3} together guarantees that the nullspace of $V$ is empty so that $V$ is invertible. It follows that $V^{-1} \geq 0$ from Theorem \ref{thm:2}, and we are done.
\end{proof}

\end{document}